\documentclass[12pt]{amsart}
\usepackage{amssymb}
\usepackage{graphics}

\def\@abssec#1{\vspace{.05in}\footnotesize \parindent .2in
{\bf #1. }\ignorespaces}

\newtheorem{lemma}{Lemma}

\newcommand{\E}{{\mathbb E}}

\renewcommand{\le}{\leqslant}
\renewcommand{\ge}{\geqslant}

\begin{document}

\title[Bubble formation and collapse]{A simple model for asset price bubble formation and collapse}
\author{Alexander Kiselev}
\thanks{Department of
Mathematics, University of Wisconsin, Madison, WI 53706; e-mail:
kiselev@math.wisc.edu}
\author{Lenya Ryzhik}
\thanks{Department of
Mathematics, Stanford University, Stanford, CA 94305; e-mail:
ryzhik@math.stanford.edu}



\begin{abstract}
We consider a simple stochastic differential equation for modeling
bubbles in social context. A prime example is bubbles in asset pricing, but similar mechanisms
may control a range of social phenomena driven by psychological factors (for example, popularity of rock groups,
or a number of students pursuing a given major). Our goal is to study the simplest
possible model in which every term has a clear meaning and
which demonstrates several key behaviors. The main factors that
enter are tendency of mean reversion to a stable value,
speculative social response triggered by trend following and random
fluctuations. The interplay of these three forces may
lead to bubble formation and collapse.
Numerical simulations show that the equation has distinct regimes depending on the values of the parameters.
We perform rigorous analysis of the weakly random regime, and study the role of change in fundamentals
in igniting the bubble.
\end{abstract}
\maketitle

\section{Introduction}

The best known and well-studied examples of social bubbles are speculative bubbles in
asset pricing.
These have a long history that in some form can be traced back at least to ancient Rome \cite{Chancellor}.
First well-documented examples of speculative bubbles are the tulip mania in Netherlands in 1637 (see e.g.
\cite{Dash,Garber,Goldgar}) and the South Sea Company bubble of 1720
\cite{Carswell,Cowles}. It is after the latter boom and collapse
that the term "bubble" was coined. Generally, economists use the
term "bubble" to describe an asset price that has risen above the
level justified by the economy fundamentals, as measured by the
discounted stream of expected future cash flows that would accrue to
the owner of the asset. In practice, of course, what makes it
difficult to diagnose a bubble is the uncertainty present in any
calculation of the future returns and of the appropriate discount rates.

Major modern examples of bubbles include the Japanese asset price
bubble of 1980's \cite{Shiratsuka} which involved both real estate
and equities, the Dot.com bubble \cite{Cassidy} in the US
information technology stocks, and the recently punctured real
estate and wider credit bubble \cite{SorW}, which was centered in the US and UK
but was present and affected a much wider range of markets. There
have been many more smaller and local bubbles, involving regional
real estate, stocks in a certain sector, and individual stocks
\cite{McKay,Kind}.

A simple intuitive explanation of the mechanism behind bubble
formation has been suggested by Shiller \cite{Shiller}: 
"If asset prices start to rise strongly, the success of some
investors attracts public attention that fuels the spread of
the enthusiasm for the market. New (often, less sophisticated) investors
enter the market and bid up prices. This "irrational exuberance"
heightens expectations of further price increases, as investors
extrapolate recent price action far into the future. The market's
meteoric rise is typically justified in the popular culture by some
superficially plausible "new era" theory that validates the
abandonment of traditional valuation metrics. But the bubble carries
the seeds of its own destruction; if prices begin to sag, pessimism
can take hold, causing some investors to exit the market. Downward
price motion begets expectations of further downward motion, and so
on, until the bottom is eventually reached".

It is reasonable to assume that bubble formation in asset
prices has its root in some basic aspects of the human social
psychology, which may manifest itself as soon as some necessary
conditions (such as existence of basic liquid markets) are
satisfied. Moreover, it is likely that bubble formation and collapse
phenomena appear in areas of social dynamics beyond the asset price
formation. One example may be attraction of an excessive number of
people to some profession, after there has been a significant
breakthrough. The mechanism of the bubble in this case can be modeled very
similarly to the price formation: explosive growth attracts newcomers
who project growth into the future.
Once the readily available new applications have been worked
through, however, the profession may suffer from overcapacity of
labor. Other relevant example can be the number of people following
certain fashion or popularity of a rock group.


The existing research literature on asset price bubble formation is too enormous
to attempt a complete overview in this article.  There are many
different approaches and directions. We only indicate a few
branches. In finance literature, much work has been devoted to finite
difference price evolution models. The first model of so-called
rational bubbles based on finite differences was proposed by
Blanchard and Watson \cite{BW1982}. The asset price satisfies a finite difference equation
which basically expresses the no-arbitrage condition. The bubble component may be present
due to the non-uniqueness of the solution. For discussion and more refined versions see also e.g.
\cite{AB,AS92,Cam89,Fuk98,KS08,LS99}. Continuous time models based on strict local martingale approach
have been considered in \cite{CH,HLW,JPS,LW1,LW2}, where further
references can be found. Another approach uses
finitely additive measures (charges) to model the pricing bubbles (see e.g. \cite{G1,GL,JM}). The
charges and local martingales approaches were recently shown to be equivalent in \cite{JPS}.
All the models discussed so far have no-arbitrage conditions explicitly built into them.

Closest to our interest here is the direction of agent-based models of price formation. The primary motivation
of this class of models is more realistic replication of the statistical properties observable in the time
series of stock prices returns. The agent-based models usually incorporate
a number of agents following different strategies with effective price equation
driven by the balance of supply and demand. The no-arbitrage conditions are typically absent,
even though for sophisticated models constructing consistently winning strategies may be challenging.
Two primary classes of strategies that received most attention are trend following ("chartist")
and fundamentalist. The dynamics generated by interaction of trend following ("chartist") trading
strategies and fundamentals-driven investors has been the focus of many works
(see e.g. \cite{baum,dhuang,D1,farmer,IS,lux}). Other authors pursued more sophisticated models based on
heterogeneous and adaptive beliefs where strategies may vary in time (see e.g. \cite{BD,BH,LM1,LM2,L,SX,GH}).
Yet another direction of research is based on parallels with statistical physics
and trading network models.  An elegant
theoretical physics-inspired model leading to superexponential
growth in prices while speculative heat lasts was proposed in
\cite{Derman1}. We refer to \cite{SZSL,Sorn1} for recent reviews of different
directions and contributions, where many more references can be found.

Our goal in this paper is to introduce and investigate a simple stochastic differential
equation which models creation of bubbles in asset pricing. The equation has three key terms.
The first is the mean reversion term, driving price back to the fundamental value. It models
the contribution of the fundamentalist trading strategy, and is similar to what has been considered before.
The second is random term, which models exogenous factors. The third term is the \it speculative
\rm or \it social response term, \rm which models psychology of trend followers.
Though the latter term is also similar in spirit to what has been considered before (see e.g. \cite{dhuang,GH,farmer,lux}),
the exact form is different in that our term is nonlinear, which is essential for certain dynamical features.

The direction of our work is different from earlier literature in several respects. First, we intentionally keep
the equation conceptually as simple as possible -- though also sufficiently complex to produce rich set of phenomena.
Thus we do not aim at this point to produce a model that can closely explain all of the observable statistical features of complex
modern markets, but rather look for the simplest signature model of bubble and collapse, perhaps the next order of
approximation to the reality after the random walk. One motivation is that even if not exhaustive, a simple model has a better chance of
being capable of calibration to the empirical data without overfitting.
There are only three independent parameters in the model, and we investigate the behavior
of the model across the possible values of the parameters. In one of the regimes -- small randomness -- the model
can be analyzed rigorously, and this provides valuable insight into the possible behavior in other regimes.
Secondly, randomness plays more important role in our model than usual. In particular, in the absence of the change in
fundamentals, randomness is entirely responsible for
igniting the bubble and causing the bubble to collapse. This means that the deterministic part of our dynamics
does not suggest any typical time scales for these processes, making them essentially random, and
similar to Poisson process. Indeed, the bubble collapse (or ignition) are notoriously hard to predict. 
Third, we focus on the influence and role of the change in fundamentals in igniting the bubble, which is arguably
a key reason behind the initiation of many bubbles. In our model's framework, a change in fundamentals may significantly
increase the likelihood of bubble ignition.

Here is the equation that we are going to study, written for the logarithm of the asset price:
\begin{equation}\label{bubble}
dP(t) = -\mu (1-e^{P_0-P(t)})dt + \sigma dB_t + \nu
S(P(t)-P(t-T))dt.
\end{equation}
The parameters $\mu,$ $\sigma$ and
$\nu$ regulate the strength of the mean reversion, random and
speculative terms respectively, 
$P_0$ is the fundamental price of the asset (that is assumed constant for
the moment but will be allowed to vary later), $B_t$
is the Brownian motion and $S(x)$ is the social response or
speculative function. It will be assumed to have the
following properties: $S(x)$ is odd, monotone increasing and $S(x)
\rightarrow 1$ as $x \rightarrow \infty.$ The function $S$ has the
natural structure reinforcing the existing trend in price dynamics.
It is important, however, that the reinforcement strength
depends on the rate of past returns in a certain way. We will assume
that $\nu\gg\mu,$ so that for large values of the argument, the speculative
function can dominate mean reversion. On the other hand, we will
assume that for small values of the argument, the speculative fever
is negligible and is dominated by mean reversion; in particular, $\nu
S'(0) \ll \mu.$

The social response term of this structure looks reasonable for
many problems where human psychology is concerned. In our opinion, it is
reasonable that the size of
the social response term varies nonlinearly depending on the
strength of the trend. Indeed, weaker trends in pricing are not as
eye catching and generate significantly less attention, news and
press coverage. The exact shape of the function $S,$ as far as it satisfied the
properties outlined above, did not affect much the qualitative properties of the
time series in our numerical experiments. One
could argue, however, in favor of a more subtle dependence of the social
response on the past trend. 
This would be equivalent to using a
more complex memory integral operator as the argument for the
social response function, rather than just $P(t)-P(t-T).$ This
more general case is interesting, and may be crucial for matching
realistic price dynamics properties of liquid financial instruments (see
the discussion section at the end of the paper for more details).
However, in the current paper we restrict
ourselves to a simple time delay for which equation \eqref{bubble}
already exhibits rich behavior.

Note that the parameters $\mu$ and $\nu$ have dimension $1/t$, and
$\sigma$ has dimension $1/\sqrt{t}.$  Without loss of generality, we
can set the delay time $T$ in \eqref{bubble} equal to one. Indeed,
any time delay $T$ can be reduced to this case by rescaling time and
other coefficients. For the rest of this paper, we will fix $T=1.$


The model \eqref{bubble} has several essential regimes,
in particular the stable mean reversion dynamics, bubble, and collapse.
These regimes are especially clear cut when the randomness is small, but become less
evident when the randomness increases.
The transition probabilities between the regimes depend on the
parameters of the model. To develop intuition, in
Section~\ref{det} we look at the possible regimes of the deterministic model
with
$\sigma =0.$ The deterministic equation
possesses the same key regimes as the random
model, but no switching between the regimes is possible without
randomness.
In Section~\ref{her} we provide an heuristic simplified picture for
the bubble equation behavior, which is useful to keep in mind when dealing
with the general case.
In Section~\ref{reg}, we look at the random case. 
Here we show that, provided the randomness is small, our basic
regimes remain stable with high probability. We also derive
estimates on the probability of switching between different regimes,
which becomes possible with randomness.
In Section~\ref{forc}, we show that the probability of bubble
creation can be enhanced greatly by manipulation of the stable value
$P_0.$ This may correspond to a strong earnings report exceeding
expectations in the case of a stock, or to a stimulative interest
rate policy in the case of the real estate market or commodities.
In Section~\ref{num}, we present some basic numerical simulations,
illustrating some of the results we prove and testing more general
parameter values. 
A more extended set of simulations, in particular studying the statistical properties of the
model (presence of fat tails for the distribution of returns? clustering of volatility? correlation
functions?) are postponed to a   later publication.
In Section~\ref{disc}, we discuss various extensions and
generalizations of our model addressing more realistic trend-following
speculation term as well as modeling bubbles in spatially distributed systems.

\section{The Deterministic Case and Regimes}\label{det}

As a first step, we look at various regimes in the deterministic bubble
equation dynamics, that is, we set $\sigma =0$.
Let us introduce the shortcut notation
\[
f(P,P_0) = -
\mu(1-e^{P_0-P(t)}).
\]
The properties of the function $S$ will be
very important, so let us recall once again that the function $S$
is odd, increasing, and $S(x) \rightarrow 1$ as $x \rightarrow
\infty.$ In addition to this, we will assume that there is no
oscillation in $S$: more precisely, the second derivative of $S(x)$
is continuous, positive on $[0,b)$ for some $b>0$ and negative for
$x \in (b,\infty).$ The symmetry (oddness) assumption is not necessary and is
made simply for convenience. The assumption of the saturation
(existence of a finite limit of $S(x)$ as $x\to\infty$,
rather than a continued growth for large values of $x$) may be
debatable. Some authors suggest that often one can discern super-exponential
growth approaching the height of a bubble \cite{Derman1,JPS,Sorn1}.

Throughout the paper, we will make certain assumptions on the
parameters and functions appearing in bubble equation \eqref{bubble}. Our first
assumption ensures this equation possesses the \it mean
reversion regime, \rm where psychology has negligible effect on price dynamics. \\

\it Assumption I. \rm $\nu S'(0) \ll \mu.$ \\

We will need the following more technical version of this assumption
to facilitate the proofs. \\

\it Assumption I'. \rm There exists $\delta_m >0$ such that
we have
\[
f(P_0+x,P_0) + \nu S(x+\delta_m)< - c_m \mu \delta_m,
\hbox{ for
$\delta_m/4 \leq x \leq \delta_m,$  }
\]
and
\[
f(P_0+x,P_0) + \nu S(x-\delta_m) > c_m \mu \delta_m,
\hbox{ for $-\delta_m/4 \geq x \geq - \delta_m,$}
\]
where $c_m$
is a fixed positive constant. \\ 

As we mentioned already in the introduction, Assumption I is
quite natural. After all, if the trend is small, there is not much
social excitement about it. It is reasonable to assume that the
speculative fever starts only when the trend is significant.

We now define the mean reversion regime and prove rigorously its
stability in the absence of randomness.

\it Mean reversion regime. \rm Assume that $|P(t)-P_0| < \delta_m$
for $t \in [t_0-1,t_0].$ Then the same holds true for any $t>t_0.$
\begin{proof}
Assume that $P(t)$ violates the corridor $|P(t)-P_0|<\delta_m$
at some point and let $t_1$
be the minimal time greater than $t_0$ when $P(t_1)-P_0 = \delta_m$
(the case $P(t_1)-P_0 =-\delta_m$ is similar). Then
\begin{equation}\label{jul61}
P'(t) = f(P_0+\delta_m,P_0) + \nu S(P(t)-P(t-1)) \leq
f(P_0+\delta_m,P_0)+ \nu S(2\delta_m) <0
\end{equation}
by definition of $t_1$ and
Assumption I'. This is a contradiction.
\end{proof}

\begin{figure}
\begin{center}
\scalebox{0.75}{\includegraphics{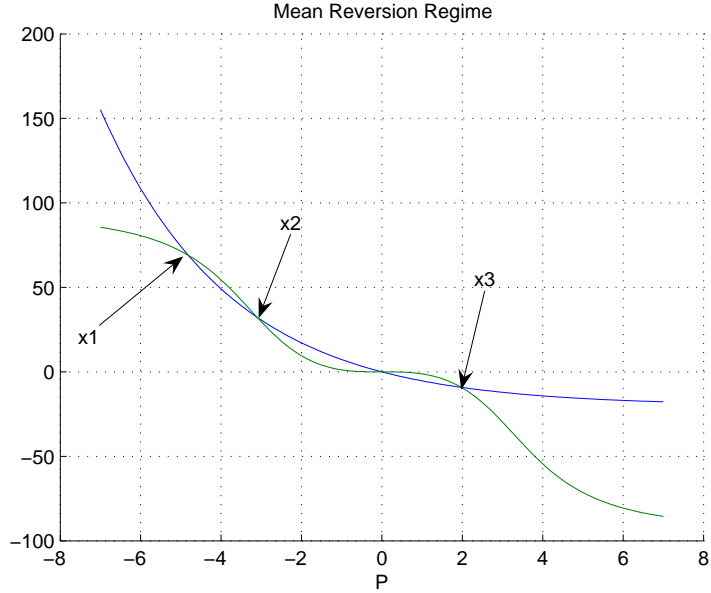}} \caption{The graphs of
$f(P,0)$ and $-\nu S(2P).$ The region between $x_2$ and $x_3$
corresponds to the mean reversion regime.} \label{mr}
\end{center}
\end{figure}

Figure~\ref{mr} illustrates the mean reversion regime (the value
of $P_0$ is set to be zero here). The region between points $x_2$ and $x_3$ on
this figure is the range of values of $P$ where the mean reversion
term dominates the social response function, as in (\ref{jul61}). The parameter
$\delta_m$ is the characteristic price scale of the mean reversion
regime. Existence of such scale satisfying Assumption I' follows
from Assumption I since $\partial f(P_0,P_0)/\partial P=-\mu.$ Below we will
sometimes consider $P_0$ depending on time. In that case, Assumption
I' holds at each $t$ with the corresponding $P_0(t).$

Our second assumption ensures existence of the bubble and collapse
regimes. \\

\it Assumption II. \rm Equation $\mu +x = \nu S(x)$ has exactly
two positive roots, $x_5$ and $x_6,$ and exactly one negative root,
$x_1.$ \\

It follows that equation $x= \nu S(x)$ has exactly two positive roots
 and exactly two negative roots, $x_2$ and
$x_3$ (due to symmetry). The smaller of these positive
roots will be denoted $x_4.$

Let us explain why Assumption II is relevant for the presence of the
bubble regime. Intuitively, a sustained bubble is driven by the
balance between social response function, growth in $P$ and the mean
reversion term. For large values of $P,$ the mean reversion term is
basically equal to $-\mu.$ Setting $P'=x,$ and approximating
 $P(t)-P(t-1)\approx P'(t)$,
we arrive at the balance
\begin{equation}\label{jul62}
x+\mu =\nu S(x)
\end{equation}
for the stable bubble regime. Given our assumptions on the
structure of the response function $S,$ equation (\ref{jul62})
can have two, one (this is a degenerate case) or zero positive roots. In the last two
cases, the nonlinear response term is just too weak to sustain a
bubble. Therefore, Assumption II is
essentially just about the sufficient strength of the social response
term. Observe that the smaller positive root, $x_5,$ is unstable,
while $x_6$ is stable (see Figure~\ref{bc1}). The root $x_6$ then
determines the asymptotic rate of growth in a bubble. The root
$x_5,$ on the other hand, determines the range of stability of the
developed bubble regime (should a random fluctuation reduce the rate
of growth below $x_5,$ the bubble will puncture and collapse will
likely begin). Let us set $\delta_b = (x_6-x_5)/2$ -- this parameter
can be regarded as the typical stability scale of the bubble regime
(it has dimension of price/time). Note that it is clear from the
definition that $x_5 > \delta_m.$

The positive root $x_4$ is relevant for values of $P$ not too far
away from $P_0.$ When the mean reversion is small, $x_4$ is the
threshold rate of growth for the start of the bubble. However, if
the rate of growth will not accelerate quickly to beyond $x_5,$ the
beginning bubble may slow down and turn with the growth of $P$ -- due to
the mean reversion term getting stronger and approaching the value $-\mu.$
\begin{figure}
\begin{center}
\scalebox{0.75}{\includegraphics{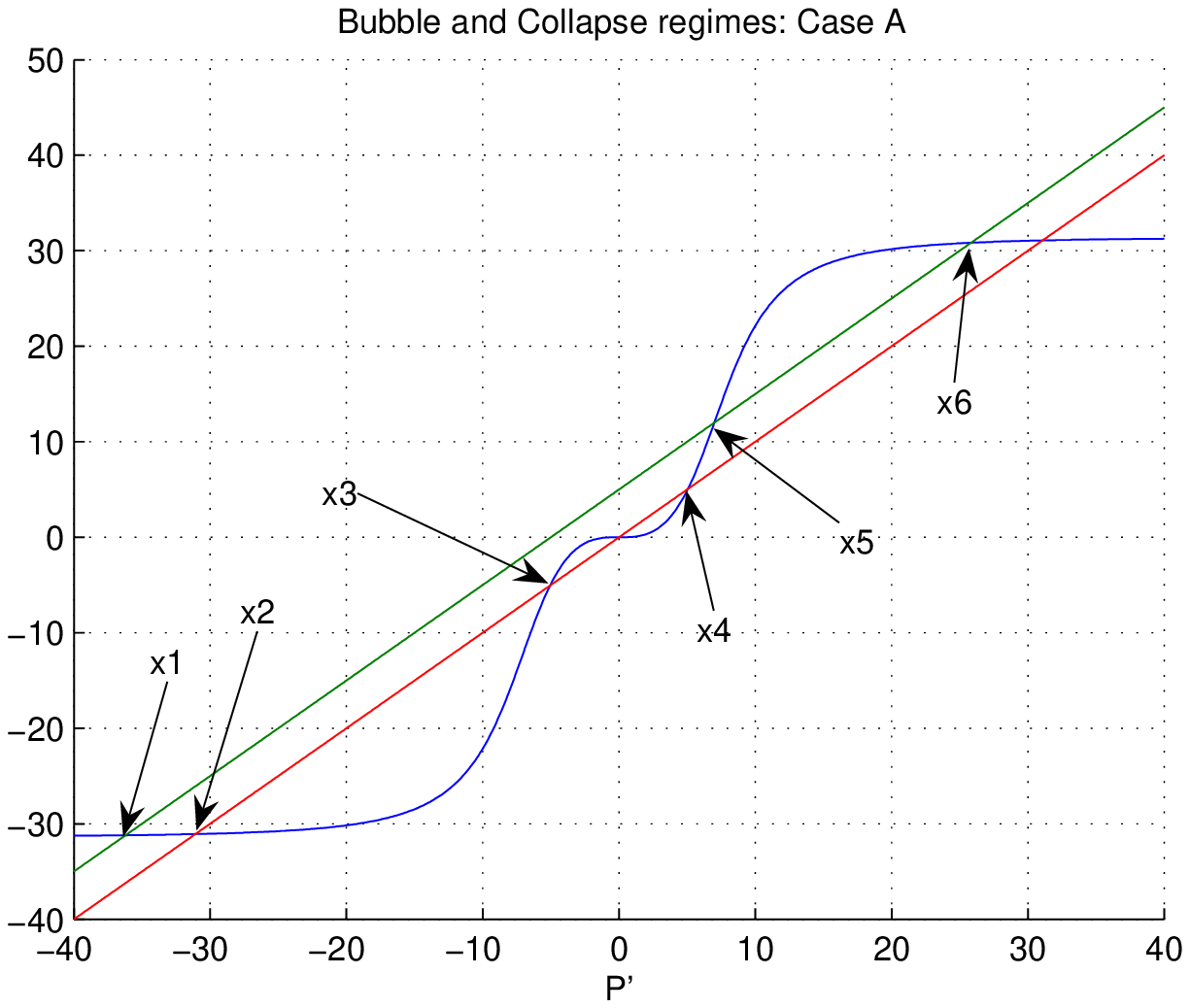}} \caption{} \label{bc1}
\end{center}
\end{figure}
Similarly to $x_6,$ the stable negative root $x_1$ determines the
rate of collapse for large values of $P.$ The collapse regime is
extremely stable for large $P.$ The system tends to go into a collapse with
the rate of decay aproaching $x_1$ if only the rate of growth
falls below~$x_5$.

Observe that without Assumption II, under our structure assumptions on the function $S$ there is a
possibility that (\ref{jul62}) has three
negative roots (or, in a degenerate case, two negative roots): see
Figure~\ref{bc2} for illustration.
In this case, the smaller
(in absolute value) stable
root $z_2$ determines the likely rate of decay right after the
puncture of a bubble. We call such scenario "a non-panic collapse",
since in many cases the primary driver of such decay is the mean
reversion force that is close to the value $-\mu$, and the social
response function $S$
does not play a big role. If $\mu$ is small, this regime may look
like simply a pause in the bubble growth. A switch either into the
bubble or into steeper collapse is possible if
random fluctuations are present. Alternatively, the "non-panic collapse" can
simply bring the price all the way down to the equilibrium value $P_0$, and
switch to the mean reversion regime. While this configuration is
interesting and adds an extra phenomenon which deserves to be
studied, in this paper we restrict ourselves to the simpler case covered by
Assumption II: most collapses in practice tend to involve panic.
The case with three negative roots will be studied elsewhere.

For values of $P$ closer to
equilibrium, where $f(P,P_0)$ cannot be regarded as just $-\mu,$ the
stable negative root $x_2$ of $x = \nu S(x)$ approximates the rate of
collapse decay. The unstable root $x_3$ provides a bound for the
range of stability of the collapse (the collapse is definitely
stable up to decay rates of $x_3;$ it really is more stable since
the mean reversion works in favor of collapse for values
$P(t)>P_0$). We set $\delta_c = (x_3-x_2)/2,$ and regard $\delta_c$
as the typical stability scale of the collapse regime. Let $a_b =
x_5 +\delta_b/2,$ and $a_c = x_3 -\delta_c/2.$ Observe that
\begin{eqnarray}\label{bubstab}
\nu S(a_b)-\mu-a_b = c_b \delta_b >0, \\
\label{colstab} \nu S(a_c) - a_c = -c_c \delta_c <0,
\end{eqnarray}
where $c_b,c_c>0$ are some fixed constants depending on $S.$
These
inequalities, which follow from Assumption II, will be
useful for us in the random case, where we will need some cushion to
ensure the stability of the bubble and collapse regimes.

\begin{figure}
\begin{center}
\scalebox{0.75}{\includegraphics{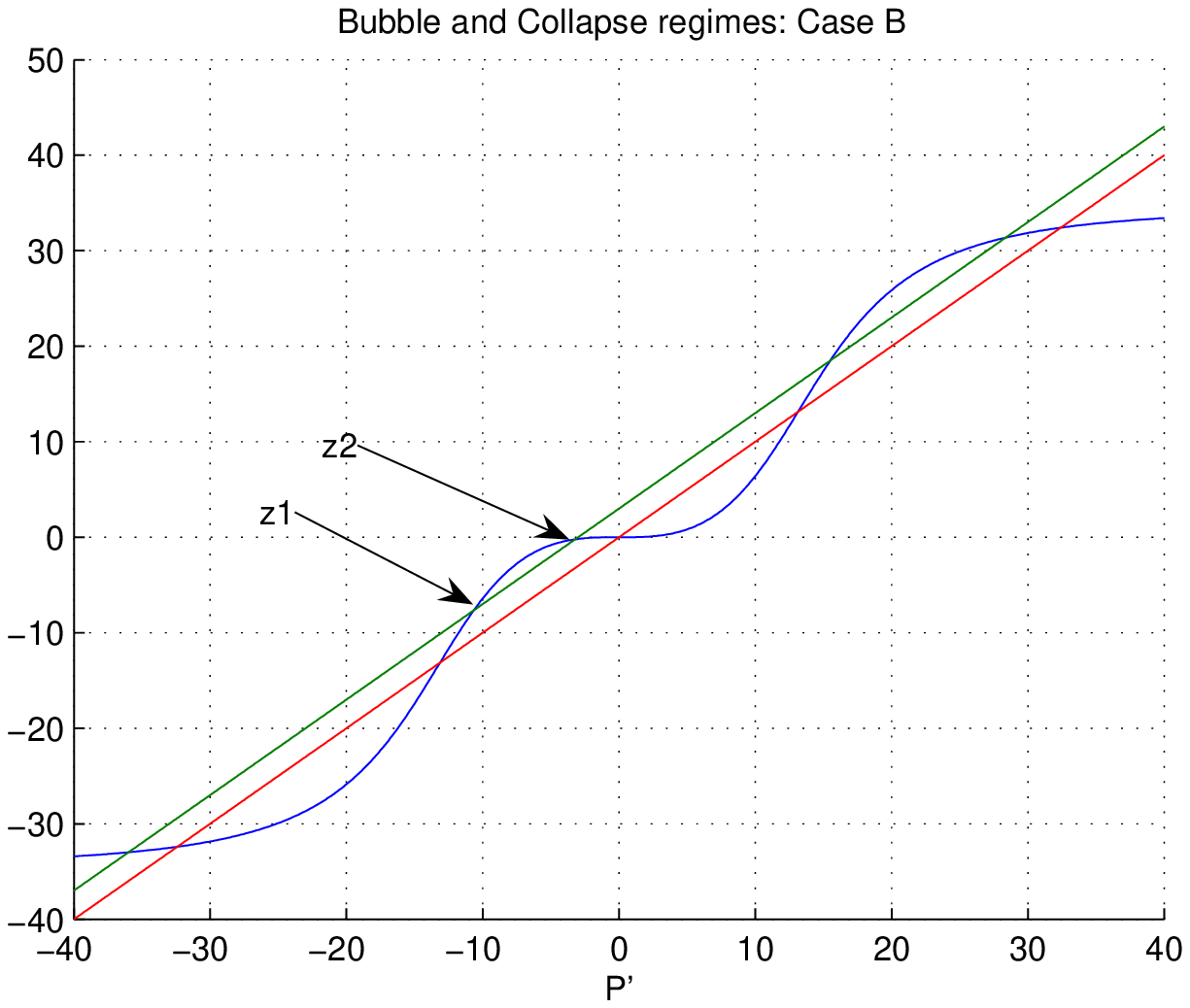}} \caption{} \label{bc2}
\end{center}
\end{figure}





We now describe the bubble and collapse regimes and their stability in
the deterministic case.

\it The bubble regime. \rm Assume that $P(t)-P(t-1) > x_5$ for all $t
\in [t_0-1,t_0].$ Then we also have $P(t)-P(t-1) > x_5$ for all
$t>t_0.$ Moreover, $P'(t) \rightarrow x_6$ as $t \rightarrow
\infty.$ \\
\begin{proof}
Assume, on the contrary, that there exists a time $t_1$ such that
$P(t_1)-P(t_1-1)=x_5,$ and $t_1$ is the smallest such time greater
than $t_0.$ However in that case, for every $t \in [t_1-1,t_1],$ we
have
\[
P'(t) \geq -\mu + \nu S(P(t)-P(t-1)) > -\mu +\nu S(x_5) = x_5
\]
due to Assumption II. It follows that $P(t_1)-P(t_1-1)>x_5$, which is a
contradiction. Hence, $P'(t)> x_5$ and
\begin{equation}\label{jul63}
\hbox{$P(t)-P(t-1)>x_5$ for all $t>t_0$.}
\end{equation}

Observe that, as a consequence, $P(t) \rightarrow \infty$ as $t
\rightarrow \infty$. Let us denote
\[
P_n = {\rm min}_{t \in [t_0+n-1,
t_0+n]} P(t),
\]
then the sequence $P_n$ is increasing by (\ref{jul63}), and, moreover,
$P_n\ge C_0+nx_5$, with some constant $C_0$ that depends only on the data
on the time interval $[t_0-1,t_0]$.
We also define
\[
p^+_n = {\rm max}_{t \in [t_0+n-1,t_0+n]}
P'(t),
\]
and
\[
p^-_n = {\rm min}_{t \in [t_0+n-1, t_0+n]} P'(t).
\]
We claim that
\begin{equation}\label{jul71}
\lim_{n\to\infty}p_n^-=\lim_{n\to\infty}p_n^+=x_6.
\end{equation}
To this end we make the following observations. First,
note that $p_n^->x_5$ and $p_n^+>x_5$ for all $n$.
We also claim that
\begin{equation}\label{jul73}
\hbox{if $p_{n}^-\le x_6$ then $p_{n+1}^-\ge p_{n}^-$}.
\end{equation}
Indeed,  assume that $p_n^-\le x_6$. For any $t\in[t_0+n,t_0+n+1]$ we have
\begin{equation}\label{jul74}
P'(t)=-\mu+\mu e^{P_0-P}+\nu S(P(t)-P(t-1))> -\mu+\nu
S(\min(p_n^{-},p_{n+1}^-)),
\end{equation}
thus
\begin{equation}\label{jul75}
p_{n+1}^-> -\mu+\nu S(\min(p_n^{-},p_{n+1}^-)).
\end{equation}
Now, if (\ref{jul73}) fails then (\ref{jul75}) implies that
\begin{equation}\label{jul76}
p_{n+1}^->-\mu+\nu S(p_{n+1}^-).
\end{equation}
As $p_{n+1}>x_5$ it follows from (\ref{jul76}) that $p_{n+1}^->x_6$, hence we
have
the inequalities
\[
x_6\ge p_{n}^-\ge p_{n+1}^->x_6,
\]
which is a contradiction. Thus, (\ref{jul73}) holds.

Next, we show that
\begin{equation}\label{jul78}
\hbox{if $p_k^-\ge x_6$ then $p_{n}^-\ge x_6$ for all $n\ge k$.}
\end{equation}
It suffices to show that (\ref{jul78}) holds for $n=k+1$. However,
it follows from
(\ref{jul75}) that either
\begin{equation}\label{jul79}
p_{k+1}^-> -\mu+\nu S(p_k^{-})\ge -\mu+\nu S(x_6)=x_6,
\end{equation}
or
\begin{equation}\label{jul79.5}
p_{k+1}^-> -\mu+\nu S(p_{k+1}^-),
\end{equation}
which also implies that $p_{k+1}^->x_6$. Therefore, (\ref{jul78}) holds.

The next step is to show that there exists a constant $c_0$  so that
\begin{equation}\label{jul710}
\hbox{ if $p_n^+\ge x_6+c_0e^{-nx_5}$ then $p_{n+1}^+\le p_n^+$.}
\end{equation}
Note that
\[
p^+_{n+1} \leq -\mu + \mu e^{P_0-P_n} + \nu S(\max(p^+_n,p_{n+1}^+))\le
-\mu + \mu e^{P_0-C_0}e^{-nx_5} + \nu S(\max(p^+_n,p_{n+1}^+)) ,
\]
hence
\begin{equation}\label{jul712}
p^+_{n+1} \leq
-\mu + d_0e^{-nx_5} + \nu S(\max(p^+_n,p_{n+1}^+)) ,
\end{equation}
with $d_0=\mu\exp(P_0-C_0)$. Now, if $p_n^+<p_{n+1}^+$ then
\[
p^+_{n+1} \leq
-\mu + d_0e^{-nx_5} + \nu S(p_{n+1}^+) ,
\]
which, in turn, implies that
\[
p_n^+<p_{n+1}^+<x_6+c_0e^{-nx_5}
\]
with some constant $c_0$ that depends only on the behavior of the
function $S(x)$ near $x=x_6$,
and (\ref{jul710}) follows.

Now, let us prove (\ref{jul71}). Let us set
\[
L^+=\limsup_{n\to\infty}p_n^+.
\]
Passing to the limit in (\ref{jul712}) gives
\[
L^+\le -\mu +\nu S(L^+),
\]
hence  $L^+\le x_6$. Then (\ref{jul73}) and (\ref{jul78}) imply that
\[
L^-=\lim_{n\to\infty}p_n^-
\]
exists and $x_5<L^-\le x_6$. Moreover, (\ref{jul75}) implies that
\[
L^-\ge-\mu+\nu S(L^-),
\]
whence $L^-=x_6$. As   $L^+\le x_6$, and $L^+\ge L^-$, it follows
that $L^+= x_6$ as well, and, moreover,
\[
\liminf_{n\to\infty}p_n^+\ge\liminf_{n\to\infty}p_n^-=x_6,
\]
thus (\ref{jul71}) holds.
\end{proof}

While the mean reversion and bubble regimes in the deterministic
case continue indefinitely, the collapse regime ends once we reach
the equilibrium value (in fact, as will be clear from the argument
below, we should go well below the equilibrium value to break the decay). \\

\it The collapse regime. \rm Assume that both $P(t)-P(t-1) < x_3$ and $P(t) >
P_0$ for all $t \in [t_0-1,t_0].$ Then for any $t>t_0$ such that $P(s)$
remains larger than $P_0$ for all $t_0 \leq s \leq t,$ we have
$P(t)-P(t-1) < x_3.$ Moreover, if $P(t)>P_0$ for $t \in
[t_0-1,t_0+1],$ and \[
x_3> {\rm max}_{t \in [t_0-1,t_0]}(P(t)-P(t-1))
\equiv d_0
> x_2,
\]
then the decay accelerates in the next time interval:
\[
{\rm max}_{t \in [t_0,t_0+1]}(P(t)-P(t-1)) \equiv d_1 < d_0.
\]
 \\
\it Remark. \rm We lack a statement on the asymptotic rate of
collapse since any collapse ends once $P_0$ is reached. But if one
assumes that $P$ is very large in the beginning, one can prove that the
rate of decay approaches $x_1$ while $P$ remains very large. The
argument is similar to that given above in the bubble case.
\begin{proof}
Suppose there exists $t_1>t_0$ such that for $t \in [t_0,t_1)$, we have
$P(t)-P(t-1) < x_3$ and $P(t)>P_0,$ but $P(t_1)-P(t_1-1) = x_3.$
Then for any $t \in [t_1-1,t_1],$  we have
\[P'(t) \leq \nu S(P(t)-P(t-1)) <
\nu S(x_3) =x_3,\] giving us a contradiction.

Now let ${\rm max}_{t \in [t_0-1,t_0]}(P(t)-P(t-1))= d_0,$
with $x_2 <
d_0 <x_3.$ Then for every $t \in [t_0-1,t_0],$ we have $P'(t) < \nu
S(d_0) < d_0$ by the definition of $x_2$ and $x_3.$ This implies, by a
familiar argument, that $P(t)-P(t-1) < d_0$ for every $t \in
[t_0,t_0+1].$
\end{proof}
It is clear that $P$ should fall sufficiently below $P_0$ to balance
the speculative function. In particular, if the fall has been from
large values of $P,$ we need at least $f(P,P_0) \gtrsim x_2$ to
break the fall.

We should stress that it is clear from our conditions that mean
reversion, bubble and collapse do not exhaust the whole range of
possible behaviors in our model. For most parameter ranges where the
system does exhibit interesting behavior, there is a significant gap
between mean reversion and bubble. We call it a \it transitory
regime. \rm In the random case, one situation where it appears is when a
fluctuation kicks the price away from $P_0$ but not strongly enough
for the bubble to form (for example, the rate of growth does not exceed
$x_4 > \delta_m$ or is just barely above it).
In this regime, there
is no dominating term. The effects of both mean reversion and social
response may be significant. In the deterministic setting, what
happens depends strongly both on the rate of growth and the
value of $P$: for
example, if $P>P_0$ and ${\rm max}_{t \in [t_0-1,t_0]}(P(t)-P(t-1))
< x_4,$ the growth will decelerate. In the random case, the
evolution in the transitory regime becomes much more sensitive to
randomness. The transitory regime is thus most difficult for
rigorous analysis. In this paper, we will prove rigorous bounds on
probabilities reflecting stability or transitions between mean
reversion, bubble and collapse. The bounds we prove may be not quite
sharp due to the transitory effects we just described.

Observe that in the deterministic case, the three basic regimes
 are not connected: one cannot transition between them, with the
exception of the collapse regime which ends once the price has
fallen enough. Other transitions will be made possible by
randomness. There is a way, however, to ignite bubbles in the
deterministic case as well -- by adjusting the level of the
equilibrium value, $P_0.$  This can be interpreted as a change in
the fundamental data: a new positive earnings report in the case of a
stock or a stimulative interest rate policy for real estate or fixed
income investments. Since in practice, such change in $P_0$ often
happens in a jump, that is what we will consider here. It is not
difficult to adjust the arguments below to the case of a continuous
dependence $P_0(t).$

The first observation is that if the jump in $P_0$ is small, then
the bubble does not ignite. \\

\it Driving deterministic bubbles: a small jump. \rm Assume that $|P(t)-P_0(t)| <
\delta_m/2$ for $t\in [t_0-1,t_0],$ $P_0(t)=P_-.$ Assume also
that $P_0$
jumps up at $t_0$ to a value $P_+$ and the size of the jump satisfies
$0<P_+-P_-<\delta_m/2.$ Then we have $P(t)-P_0(t)<\delta_m/2$ and
$P(t)-P_0(t)>-\delta_m$ for $t \in [t_0, t_0+1],$ $P_0(t)=P_+.$ Thus
we continue in the mean reversion regime for all future times.\\
\it Remark. \rm We need a more relaxed lower bound above simply because
the original bound may clearly fail right after the jump.
\begin{proof}
Assume on the contrary that $P(t_1)=P_+ +\delta_m/2$ at some $t_1$ for
the first time. Then we have
\[
P'(t_1) = f(P_++\delta_m/2, P_+) + \nu
S(P(t_1)-P(t_1-1)),
\]
and, moreover,
\[
P(t_1)-P(t_1-1) \leq P(t_1)-P_+ + P_+ - P_- +
P_- - P(t_1-1) \leq \delta_m/2 + \delta_m/2 + \delta_m/2 =
3\delta_m/2.\]
Using Assumption I' with $x = \delta_m/2,$ gives
\[
f(P_++\delta_m/2,P_+) + \nu S(3\delta_m/2)<0,\]
and so $P'(t_1)<0$
which is a contradiction. Similarly, assume that
$P(t_2) = P_+ -\delta_m$ for some
$t_2>t_0$ for the first time. But then
\begin{eqnarray*}
&&P'(t_2) =f(P_+-\delta_m,P_+)+\nu S(P_+-\delta_m- P(t_2-1))\\
 &&\,\,\,\,\,\,\,\,\,\,\,\, \,\,\,\,\,\,\geq
f(P_+-\delta_m,P_+) + \nu S(P_+-\delta_m- P_- -\delta_m/2) >0
\end{eqnarray*}
by
Assumption I'.
\end{proof}

On the other hand, a strong driving will necessarily lead to bubble
ignition. \\

\it Driving deterministic bubbles: a large jump. \rm Assume that $|P(t)-P_0(t)| <
\delta_m/2$ for $t \in [t_0-1, t_0],$ $P_0(t)=P_-.$ Assume that
$P_0$ jumps up at $t_0$ to a value $P_+$ and the size of the jump satisfies
\[ P_+ -P_- \geq x_5+\frac{x_5}{\mu}+\delta_m.\]
 Then for all $t \in [t_0+1, t_0+2]$ we have
$P(t)-P(t-1)>x_5.$
\begin{proof}
First, observe that for all $t \in [t_0,t_0+1]$ we have $P(t)-P(t-1)
> -\delta_m.$ This is true since $P(t-1) \leq P_-+\delta_m/2$ by
assumption, and $P(t) \geq P_--\delta_m/2$ by the same argument as
in the proof of stability of mean reversion regime (using $P_+>P_-$).

Now we claim that $P'(t) > x_5,$ for every $t \in [t_0,t_0+1].$
Indeed, if $P(t) < P_-+x_5+\delta_m/2,$ then \[ P'(t) > f(P(t),P_+)
+ \nu S(-\delta_m) = \mu(e^{P_+-P(t)}-1) - \mu(e^{\delta_m/2}-1).
\] We used Assumption I' in the last step. Then
\[ P'(t) > \mu(P_+-P(t)-\delta_m/2) \geq x_5. \]
On the other hand, if $P(t) \geq P_-+x_5+\delta_m/2,$ then
\[ P'(t) > -\mu +\nu S(P(t)-P_--\delta_m/2) \geq -\mu +\nu
S(x_5) \geq x_5. \]

Now we can prove that $P(t)-P(t-1)>x_5$ for all $t \in [t_0+1,
t_0+2]$ by a familiar argument. Indeed, this is true for $t=t_0.$
Assuming existence of some minimal $t_1$ where we have
$P(t_1)-P(t_1-1)=x_5,$ we quickly obtain a contradiction.
\end{proof}

\section{The random case: heuristics}\label{her}

In this section, we will provide an intuitive picture of what to
expect from our model~\eqref{bubble} when the random term is included, but is small.
This regime may be more relevant to modeling illiquid investments like real estate
rather than liquid ones.
Consider first the mean
reversion regime. In this case we can regard the  influence of the
social response
term  as minor. It gets activated only once randomness, by
accident, moves the price away from the mean reversion zone. Note
that if we gradually drift away from $P_0$ -- which is difficult due
to mean reversion -- we still do not activate the bubble regime as the social
response term is small, but
move into the transitory regime. Rather, we should have a jump of size
$\sim a_b$ to enter the stable bubble regime. Let us set $P_n =
P(t_0+n),$ $\Delta P_n = P_n-P_{n-1},$ $\Delta B_n =B_n-B_{n-1}.$
Then heuristically, the evolution of $P_n$ in mean reversion regime
can be approximated by
\[ \Delta P_n =  \sigma \Delta B_n - \mu (P_n-P_0). \]
Intuitively, we pass into the stable bubble regime on the
$n$th step if $\sigma\Delta B_n \gtrsim (1+\mu)a_b.$ Thus the
probability $\lambda_1$ of bubble ignition at each step is roughly
$1-\Phi(\frac{(1+\mu)a_b}{\sigma}),$ where $\Phi$ is the cumulative
function of the normal distribution,
\[ \Phi(x) = \frac{1}{\sqrt{2\pi}}\int_{-\infty}^x e^{-y^2/2}\,dy.
\]
At this heuristic level, the dynamics in the mean reversion regime
can be thought of as a Poisson-type process where on each time step
we pass to the bubble regime with probability $\lambda_1,$ and stay
in the mean reversion regime with probability $1-\lambda_1.$ Of
course, in reality the probability of bubble ignition depends on the
past dynamics and on what exactly happens within the time interval.
The social response term also increasingly plays a role as the
system moves into the transitory regime which is more prone to creating
bubbles than mean reversion. However, the above picture captures the
most qualitative aspects of dynamics.

Once we are in the bubble regime, our rate of growth approaches
$x_6,$ and the bubble stability in terms of variations in the rate
of growth is measured by $\delta_b.$ The dynamics is qualitatively
described by \begin{equation}\label{bubbleher} \Delta P_n = -\mu +
\sigma \Delta B_n + \nu S(\Delta P_{n-1}).
\end{equation} Now, in order for the bubble to burst and for dynamics
to pass to the collapse regime, we need $\sigma \Delta B_n \gtrsim
\delta_b.$ Thus, the dynamics in the bubble regime can be modeled by
a Poisson-type process where on every time step we switch from bubble
to collapse
with the probability $\lambda_2 \sim 1-\Phi(\delta_b/\sigma)$ and we
stay in the bubble with the probability $1-\lambda_2.$ Observe the
similarity between this heuristic picture and the well known simple
discrete model proposed and investigated by Blanchard and Watson
\cite{BW1982} (minus the no-arbitrage condition, which is not automatic
in our model).

Finally, once we have entered the collapse regime, the dynamics in the
leading order is described again by \eqref{bubbleher} (for larger
values of $P,$ with the term $(-\mu)$ on the right side gradually
disappearing as $P$ decreases). For large $P,$ the only way to break
the collapse is to go back into the bubble regime, and for that one
needs $\sigma \Delta B_n \gtrsim a_b-a_c.$ Observe that the
probability of this is small compared to $\lambda_1$ and
$\lambda_2,$ due to $a_b
> \delta_b,$ $a_c < -\delta_c.$ As we approach the equilibrium value
$P_0,$ the condition for breaking the collapse gradually approaches
$\sigma \Delta B_n \gtrsim \delta_c,$ and the passage to the transitory
or mean reversion regime becomes possible. Once we have broken
through $P_0,$ the push back to the mean generated by the mean
reversion term increases
rapidly. We pass from collapse to the recovery when $P_0- P \sim
\log (1+\frac{\delta_c}{\mu}).$

We emphasize that the heuristic picture of this section cannot be
expected to provide a close approximation of the real evolution in
all possible relevant ranges of parameter values. For example, when
the random term is sufficiently strong, 
transitory effects play an important role, and dependence of the
evolution on details of the past dynamics increases notably. This
regime is perhaps most interesting from the practical point of view.
Still, in Section~\ref{num} devoted to numerical simulations, we
will see that our heuristic picture appears to describe the evolution pretty
well for a significant range of parameter values.

\section{The Random Case: Mean Reversion, Bubble and Collapse}\label{reg}

Now, we consider the full equation \eqref{bubble}:
\[
dP = f(P)dt + \sigma dB_t + \nu S(P(t)-P(t-1))dt.
\]
Naturally, depending on the relative strength of parameters, this
equation can exhibit very different behavior. For example, if both
$\nu$ and $\sigma$ are small, the mean reversion term will
always dominate and we will have a process very close to
Ornstein-Uhlenbeck process. Taking $\sigma$ large emphasizes randomness, and
then the nonlinear social response effects may be difficult to discern.
In this section we provide rigorous analysis which is most relevant in the case
of small to moderate random forcing.
First, we describe the probabilistic stability of
our basic regimes. For now, the equilibrium value  $P_0$ is fixed. \\

\it Stability of the mean reversion regime. \rm Assume that for all $t
\in [t_0-1, t_0]$ we have $|P(t)-P_0| < \delta_m,$ and $|P(t_0)-P_0|
< \delta_m/2.$ Then with probability at least
$2\Phi(C\delta_m/\sigma)-1,$ we have that the same is true for $t
\in [t_0,t_0+1]:$ $|P(t)-P_0|<
\delta_m$ for $t \in [t_0,t_0+1]$ and $|P(t_0+1)-P_0|<\delta_m/2.$
Here $C>0$ is a universal constant.
\\
\it Remark. \rm A stronger condition at the end of the interval is
necessary: if we had $|P(t_0)-P_0| = \delta_m,$ then we would have
exited the
interval $[P_0-\delta_m, P_0+\delta_m]$ at later times with probability one.
\begin{proof} Let us
assume that for any $t_1<t_2$, $t_1,t_2 \in [t_0,t_0+1]$ we have
\begin{equation}\label{bmcon1}
\sigma|B_{t_2}-B_{t_1}| < {\rm min}(c_m\mu,1/4)\delta_m,
\end{equation}
where $c_m$ is the same as in Assumption I'.
This is true
in particular if
\[
\sigma {\rm max}_{t_0 \leq t \leq t_0+1}
|B_t-B_{t_0}| < \frac12 {\rm min}(\mu c_m,1/4)\delta_m.
\]
The probability
of this event is equal to $2\Phi(\frac{C\delta_m}{\sigma})-1.$

We claim that
in this case, $|P(t)-P_0| \leq \delta_m$ for all $t \in [t_0,
t_0+1]$ and $|P(t_0+1)-P_0|<\delta_m/2.$
Let us verify that $P(t) \leq
P_0+\delta_m$ for all $t \in [t_0,t_0+1]$ and $P(t_0+1)<P_0+\delta_m/2.$
The other part of the condition can be verified similarly.
Indeed, assume on the contrary that there exists $t_3$ with
$P(t_3) = P_0 +\delta_m.$ Let $t_4$ be the largest time in
$[t_0,t_1]$ such that $P(t_4) = P_0 + \delta_m/2.$ Then we have
\[ P(t_3) - P(t_4) = \int\limits_{t_4}^{t_3}(f(P(s),P_0)+\nu
S(P(s)-P(s-1))) \,ds + \sigma(B_{t_3}-B_{t_4}). \] Note that due to
$P(s-1)>P_0-\delta_m$ and Assumption I', we have
\[
f(P(s),P_0)+\nu
S(P(s)-P(s-1)) < -c_m \mu \delta_m
\]
for any $s$ in the interval of
integration. Therefore we get
\[
P(t_3) - P(t_4) < -c_m \mu \delta_m
(t_3-t_4) + \delta_m/4 < \delta_m/4,
\]
a contradiction to our
choice of $t_3,t_4.$

Similarly, under condition \eqref{bmcon1} we have
$P(t_0+1)<P_0+\delta_m/2.$ Indeed, assume that $P(t_0+1) \geq
P_0+\delta_m/2.$ If we ever have $P(t)<P_0 + \delta_m/4$ for $t\in
[t_0,t_0+1],$ let $t_3$ be the largest time less than $t_0+1$ when
$P(t_3)=P_0+\delta_m/4.$ Then
\begin{eqnarray*}
&&P(t_0+1) - P(t_3) =
\int\limits_{t_3}^{t_0+1}(f(P(s),P_0)+\nu S(P(s)-P(s-1))) \,ds +
\sigma(B_{t_0+1}-B_{t_3}) \\
&&\,\,\,\,\,\,\,\,\,\,\,\,\,\,\,\,\,\,\,\,\,\,\,\,\,\,\,
\,\,\,\,\,\,\,\,\,\,\,\,\,\,\,\,\,\,\,\,
< -c_m\mu\delta_m(t_0+1-t_3) +
\frac{\delta_m}{4}<\frac{\delta_m}{4},
\end{eqnarray*}
which is a contradiction. We used Assumption
I' in estimating the integral. On the other hand, if $P(t) \geq
P_0+\delta_m/4$ for all $t \in [t_0,t_0+1]$ (but, as we saw above,
$P(t)<P_0+\delta_m$), then
\begin{eqnarray*}
 &&P(t_0+1)-P(t_0) = \int\limits_{t_0}^{t_0+1}(f(P(s),P_0)+\nu
S(P(s)-P(s-1))) \,ds + \sigma(B_{t_0+1}-B_{t_0}) \\
&&\,\,\,\,\,\,\,\,\,\,\,\,\,\,\,\,\,\,\,\,\,\,\,\,\,\,\,
\,\,\,\,\,\,\,\,< -c_m\mu\delta_m +
c_m\delta_m \mu=0,
\end{eqnarray*}
once again a contradiction.

Thus, under our condition on the maximum of Brownian motion, the mean
reversion regime is preserved.
\end{proof}


\it Stability of the bubble regime. \rm Assume that for $t \in
[t_0-1, t_0],$ we have $P(t)-P(t-1) >a_b,$ and $\sigma {\rm max}_{t
\in [t_0-1,t_0]}|B_t-B_{t_0-1}|< c_b\delta_b/4.$ Then with
probability at least $2\Phi(C\delta_b/\sigma)-1,$ the same is true
for $t \in [t_0,t_0+1]$ (but $t_0-1$ is replaced by $t_0$).
\begin{proof}
Suppose that for any $t_1<t_2,$ $t_1,t_2 \in [t_0,t_0+1]$ we have
$\sigma(B_{t_2}-B_{t_1}) > -c_b\delta_b/2.$ This is true with
probability ta least $2\Phi(C\delta_b/\sigma)-1.$ Let us show that
in this case,
\[
\hbox{$P(t)-P(t-1)>a_b$ for all $t \in [t_0,t_0+1].$}
\]
This
holds for $t_0$ by our assumption. Let $t_1$ be the first time when
$P(t_1)-P(t_1-1)=a_b.$ Then
\[ P(t_1)-P(t_1-1) \geq -\mu +\nu S(a_b) +\sigma (B_{t_1}-B_{t_1-1})>
a_b
\]
under the assumption on Brownian motion and \eqref{bubstab}. This is a
contradiction.
\end{proof}
\it Remark. \rm We need an extra assumption on the behavior of
Brownian motion during the past interval $[t_0-1,t_0]$ since there are some
unlikely scenarios where, due to Brownian motion, we essentially set up the exit from the bubble
regime during $[t_0-1,t_0]$ even though $P(t)-P(t-1)$ stays large for this time interval.

\it Stability of the Collapse regime. \rm Assume that
$P(t)-P(t-1)<a_c$ and $P(t)>P_0$ for $t \in [t_0-1,t_0].$ Assume
also that $\sigma {\rm max}_{t \in [t_0-1,t_0]}|(B_t-B_{t_0-1}| \geq
c_c \delta_c/4.$ Then   with probability at
least $2\Phi(C\delta_c/\sigma)-1$
either there exists a time $t \in [t_0, t_0+1]$ so that
$P(t)<P_0$ ,or $P(t)-P(t-1) < a_c$  for all $t \in [t_0, t_0+1]$.
\begin{proof}
Let us assume that for any $t_1<t_2,$ $t_1,t_2 \in [t_0,t_0+1]$ we
have $\sigma(B_{t_2}-B_{t_1}) < c_c \delta_c/2.$ This is true with
probability exceeding $2\Phi(C\delta_c/\sigma)-1.$ Assume, on the
contrary, that $P(t)>P_0$ for all $t \in [t_0,t_0+1],$ but there
exists $t_1$ such that $P(t_1)-P(t_1-1) = a_c.$ But then
\[
P(t_1)-P(t_1-1) \leq \nu S(a_c) + \sigma (B_{t_1}-B_{t_1-1}) \leq
a_c-c_c\delta_c +c_c\delta_c/2 < a_c,
\]
which is a contradiction.
\end{proof}
As in the previous case, it is easy to construct a scenario showing
that we need to assume something about the Brownian motion on the
interval $[t_0-1,t_0]$ in order to obtain a reasonable bound on the
probability of the continuation of the collapse regime.\\

Next, we look at a new phenomenon which appears in our model due to
randomness.

\it The bubble ignition probability. \rm Assume that for $t \in
[t_0-1, t_0],$ we have $|P(t)-P_0| < \delta_m,$ and $|P(t_0)-P_0| <
\delta_m/2.$ Then with a positive probability exceeding
\begin{equation}\label{igpro}
p_0=e^{-(a_b+2\mu+\delta_m+c_b\delta_b)^2/\sigma^2}
\left(2\Phi\left(
\frac{{\rm min}(\delta_m,c_b\delta_b)}{4\sigma}\right)-1\right)
\left( 2\Phi\left(\frac{c_b\delta_b}{8\sigma}\right)-1\right),
\end{equation}
we have $P(t)-P(t-1)>a_b$ for all $t \in [t_0+1, t_0+2].$ Here $c_b$
is as in \eqref{bubstab}. \\
\begin{proof}
We will identify just one scenario leading to bubble ignition, and
will estimate its probability. Assume that
\begin{equation}\label{linig}
|\sigma(B_t - B_{t_0}) - (a_b+2\mu+\delta_m+c_b\delta_b)(t-t_0)|
\leq \frac{{\rm min}(\delta_m, c_b\delta_b)}{4},
\end{equation}
for all $t \in [t_0,t_0+1]$ (one sample path satisfying this bound
appears on Figure~\ref{bc3}).

\begin{figure}
\begin{center}
\scalebox{0.75}{\includegraphics{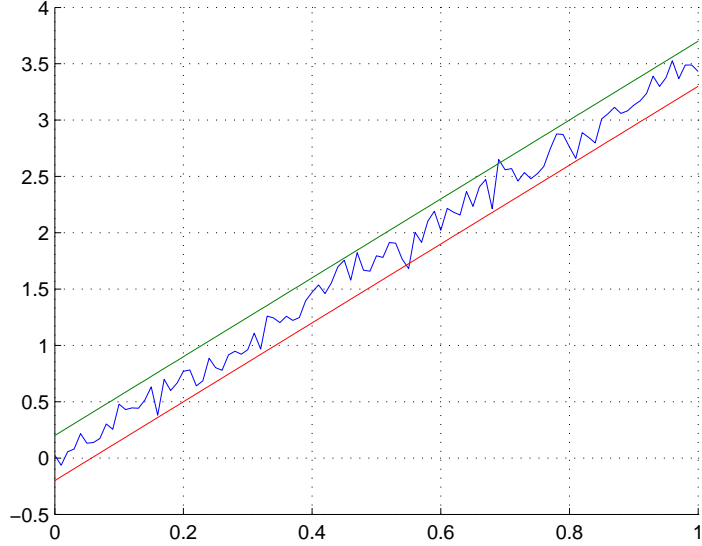}} \caption{One possible
behavior of the Brownian motion leading to the bubble ignition}
\label{bc3}
\end{center}
\end{figure}

We will estimate the probability of \eqref{linig} in
Lemma~\ref{girlemma} below. First, we claim that under assumption
\eqref{linig}, we have $P(t)-P(t-1)>-2\delta_m$ for all $t \in
[t_0,t_0+1].$ This follows from $P(t-1) \leq P_0 +\delta_m$ due to
our assumptions, and from $P(t) \geq P_0-\delta_m.$ The latter
inequality can be established, given \eqref{linig}, similarly to the
proof of the stability of random mean reversion regime.

Secondly, we claim that $P(t_0+1)-P(t_0)>a_b.$ Indeed,
\[ P(t_0+1)-P(t_0) >  - \mu + \sigma(B_{t_0+1}-B_{t_0}) +\nu
S(-2\delta_m).\] Observe that due to Assumption I', $\nu
S(-2\delta_m)=-\nu S(2\delta_m) > -\mu.$ Then using \eqref{linig}, we
obtain
\[ P(t_0+1)-P(t_0) > -2\mu
+(a_b+2\mu+\delta_m+c_b\delta_b)-\frac{{\rm
min}(\delta_m,c_b\delta_b)}{2} > a_b. \]

Finally, assume that $P(t_1)-P(t_1-1)=a_b$ for the first time at
some $t_1 \in [t_0+1,t_0+2].$ Let us split
\[
P(t_1)-P(t_1-1) = P(t_1)-P(t_0+1)+P(t_0+1)-P(t_1-1).
\]
Then
\begin{eqnarray*}
P(t_1)-P(t_0+1) \geq -\mu(t_1-t_0-1)+\sigma
(B_{t_1}-B_{t_0+1}) +\nu S(a_b)(t_1-t_0-1) \\  >
(a_b+c_b\delta_b)(t_1-t_0-1)+\sigma (B_{t_1}-B_{t_0+1}).
\end{eqnarray*} We used \eqref{bubstab} in the second step. Assume that
for any $t_2<t_3,$ $t_2,t_3 \in [t_0+1,t_0+2]$ we have
$\sigma(B_{t_3}-B_{t_2})>-c_b\delta_b/4.$ This happens with
probability exceeding
$2\Phi\left(\frac{c_b\delta_b}{8\sigma}\right)-1.$ Then
\begin{equation}\label{aux11} P(t_1)-P(t_0+1)
\geq (a_b+c_b\delta_b)(t_1-t_0-1)-c_b\delta_b/4.
\end{equation}
Next,
\begin{eqnarray} \nonumber P(t_0+1)-P(t_1-1) \geq -\mu(t_0+2-t_1) +
\sigma(B_{t_0+1}-B_{t_1-1})+ \nu \int\limits_{t_1-1}^{t_0+1}
S(P(t)-P(t-1))\,dt \\ \nonumber
 \geq -2\mu(t_0+2-t_1)+
(2\mu+a_b+\delta_m+c_b\delta_b)(t_0+2-t_1)-
\frac{{\rm min}(\delta_m,c_b\delta_b)}{2} \geq \\
 \label{aux12}
 \geq (a_b+\delta_m+c_b\delta_b)(t_0+2-t_1)-
\frac{{\rm min}(\delta_m,c_b\delta_b)}{2}.
\end{eqnarray}
Adding up \eqref{aux11} and \eqref{aux12}, we obtain
\[
P(t_1)-P(t_1-1) > a_b+c_b\delta_b/4 >a_b,
\]
and this is a contradiction.

It remains to estimate the probability of \eqref{linig}.
\begin{lemma}\label{girlemma}
Let $X_t= B_t -ct.$ Then \begin{equation}\label{lineprob} P({\rm
max}_{0 \leq s \leq t}|X_s| \leq b) \geq e^{-cb-\frac{c^2
t}{2}}(2\Phi(b/\sqrt{t})-1).
\end{equation}
\end{lemma}
\begin{proof}
The proof is a simple application of the Girsanov theorem. Let $A$ be
the set of paths where ${\rm max}_{0 \leq s \leq t} |X_s| \leq b.$
Let $Z_t = e^{cB_t - \frac{c^2 t}{2}},$ and define a new measure
$dQ_t = Z_t dP_t.$ Then by Girsanov's theorem, $X_t$ is a Brownian
motion with respect to $dQ_t$, hence $Q(A) \geq
2\Phi(b/\sqrt{t})-1.$  Denote by $\chi_A$ the characteristic function
of this set, and by $\E_Q$ the expectation taken with respect to the
measure $dQ.$ Then
\[ P(A) = \E_Q [e^{-cB_t+\frac{c^2 t}{2}}\chi_A] \geq
\E_Q [e^{-c(ct+b)+\frac{c^2 t}{2}}\chi_A]= e^{-cb-\frac{c^2
t}{2}}Q(A). \] We used that on $A,$ $B_t \leq ct+b.$ Thus,  \[ P(A)
\geq e^{-cb-\frac{c^2 t}{2}}\left(2\Phi(b/\sqrt{t})-1\right).
\]
\end{proof}
It follows from Lemma~\ref{girlemma} that the probability of
\eqref{linig} satisfies
\begin{eqnarray*}
 P \left( {\rm max}_{0 \leq s\leq 1} \left|B_s -
\frac{a_b+2\mu+\delta_m+c_b\delta_b}{\sigma}s \right| \leq
\frac{{\rm min}(\delta_m,c_b\delta_b)}{4\sigma} \right) \geq \\
e^{-(a_b+2\mu+\delta_m+c_b\delta_b)^2/\sigma^2}\left(2\Phi\left(\frac{{\rm
min}(\delta_m,c_b\delta_b)}{4\sigma}\right)-1\right)
\end{eqnarray*}
(observe that the $cb$ term in the exponent in \eqref{lineprob} is
dominated by $c^2t/2$ in our case). Multiplying it by the
probability of an independent Brownian motion behavior we required
on $[t_0+1,t_0+2]$ leads to \eqref{igpro}.
\end{proof}

Next, we look at the probability of the bubble regime switching to
a collapse.

\it The bubble collapse probability. \rm Assume that we are in a
Bubble regime: for $t \in [t_0-1,t_0],$ $P(t)-P(t-1)>a_b$ and $P(t)
> P_0.$ Then with probability exceeding
\begin{equation}\label{collpro}
e^{-(\nu-a_c+c_c\delta_c)^2/\sigma^2}\left(
2\Phi\left(\frac{c_c\delta_c}{8\sigma}\right)-1\right)
\left(2\Phi\left(\frac{{\rm
min}(\delta_m,c_b\delta_b)}{4\sigma}\right)-1\right)
\end{equation}
 we have either $P(t) \leq
P_0$ for some $t \in [t_0, t_0+2],$ or $P(t)-P(t-1) < a_c$ for all
$t \in [t_0+1, t_0+2].$
\begin{proof}
The argument is similar to the bubble ignition case, and we leave it to the
interested reader to complete it.
\end{proof}

One of the debatable features of the bubble equation is that the
bubble to collapse transition probability does not appear to
increase with growth of $P.$ In fact, there is a slight increase due
to decay in $f(P,P_0),$ but it is likely that the effect is
insignificant when $\nu \gg\mu.$ In the discussion section we
propose a slightly more complex model that is much more sensitive to
the level of $P.$ The drawback is the presence of an additional
parameter controlling the negative feedback response to growth in
price. Given how hard this parameter is to estimate in practice, we
believe there is some merit in keeping things simple and elegant
as in the bubble equation, where the timing of collapse is purely
random.

\section{The Random Case: Forcing a Bubble}\label{forc}

In this section, we consider the random case with a changing $P_0.$

\it Driving random bubbles: a small jump. \rm Assume that for all
$t \in [t_0-1, t_0],$ we have $|P(t)-P_0(t)| < \delta_m,$
$|P(t_0)-P_0(t)|<\delta_m/2$ where $P_0(t)=P_-.$ Assume that at
$t_0$ the stable value $P_0$ undergoes a jump up to $P_+,$ with
$P_+-P_-<\delta_m/4$, and $P_0(t)$ stays equal to $P_+$ after that. Then
with probability at least $2\Phi(C\delta_m/\sigma)-1,$ we have
$|P(t)-P_+|<\delta_m$ for $t \in [t_0,t_0+1],$ and
$|P(t_0+1)-P_+|<\delta_m/2.$
\begin{proof}
Let us estimate the probability of $P(t) \geq P_+-\delta_m$ for $t
\in [t_0,t_0+1]$ and $P(t_0+1) \geq P_+-\delta_m/2.$ The other half
of the statement is similar and is in fact a bit easier since there
is more danger in violating the lower boundary due to the increase
in $P_0.$ Similarly to the stability argument without change in
$P_0,$ assume that $\sigma|B_{t_3}-B_{t_2}| < {\rm
min}(c_m\mu,1/4)\delta_m$ for any $t_3>t_2,$ $t_2,t_3 \in
[t_0,t_0+1].$ The probability of this exceeds
$2\Phi\left(\frac{\delta_m}{8\sigma}\right)-1.$ Assume that there
exists $t_1$ such that $P(t_1)-P_+ = -\delta_m,$ and look at the
minimal such $t_1.$ Since $P(t_0)-P_+ > -3\delta_m/4,$ find the
largest $t_4 \in [t_0,t_1]$ such that $P(t_4)-P_+ = -3\delta_m/4.$
Consider
\begin{equation}\label{aux3} P(t_1)-P(t_4) =
\int\limits_{t_4}^{t_1} f(P(t),P_+)\,dt
+ \nu \int\limits_{t_4}^{t_1} S(P(t)-P(t-1))\,dt +
\sigma(B_{t_1}-B_{t_4}). \end{equation} Observe that \[ P(t)-P(t-1)
\geq P(t)-P_++P_+-P_--\delta_m/2 \geq P(t)-P_+ - 3\delta_m/4. \]
Then by Assumption I', we have
\[ f(P(t),P_+) + \nu S(P(t)-P(t-1)) \geq c_m\mu\delta_m \]
in \eqref{aux3}. Given our Brownian motion assumption, we obtain
\[ P(t_1)-P(t_4) \geq -\delta_m/4 + c_m\mu\delta_m(t_1-t_4)>-\delta_m/4, \]
a contradiction.

Next, assume that $P(t_0+1)<P_+-\delta_m/2.$ If we ever have
$P(t)>P_+-\delta_m/4$ in $[t_0,t_0+1],$ let $t_5$ be the largest
time where we have $P(t_5)=P_+-\delta_m/4.$ Then \begin{eqnarray*}
P(t_0+1)-P(t_5) =
\int_{t_5}^{t_0+1}(f(P(s),P_+) +\nu S(P(s)-P(s-1)))\,ds +\\
\sigma(B_{t_0+1}-B_{t_5}) \geq
c_m\mu\delta_m(t_0+1-t_5)-\delta_m/4,\end{eqnarray*} giving a
contradiction. We used Assumption I' in the second step. In the case
if $P(t)<P_+-\delta_m/4$ (but we saw $P(t)>P_+-\delta_m$) for all $t
\in [t_0,t_0+1],$ we have
\[ P(t_0+1)-P(t_0) \geq c_m\mu\delta_m - c_m\mu\delta_m >0, \]
contradiction.
\end{proof}

Now we consider the probability of bubble ignition given a
strong jump in $P_0.$ \\
\it Driving random bubbles: a large jump. \rm Assume that for all times
$t \in
[t_0-1,t_0]$ we have $|P(t)-P_0(t)|<\delta_m,$ and
$|P(t_0)-P_0(t_0)|<\delta_m/2,$ with $P_0(t) = P_-.$ Assume that at
$t=t_0$ $P_0$ jumps to a value $P_+,$ and the size of jump satisfies
\begin{equation}\label{bjump}
P_+-P_- \geq \frac{2a_b+3\delta_m}{{\rm min}(\mu,1)}.
\end{equation}
Then with probability exceeding $2\Phi(C{\rm min}(\delta_m,
c_b\delta_b)/\sigma)-1,$ we have $P(t)-P(t-1)>a_b$ for all $t \in
[t_0+1,t_0+2].$
\begin{proof}
The main step in the proof is the following claim. Suppose, in
addition to the assumptions above, that for any $t_2,t_3 \in
[t_0,t_0+1],$ $t_3>t_2,$ we have
\begin{equation}\label{bmass14}
\sigma(B_{t_3}-B_{t_2}) > -{\rm min}(\delta_m, c_b\delta_b)/2.
\end{equation}
Then
\begin{equation}\label{keybigj}
P(t_3)-P(t_2) > (a_b +{\rm min}(\delta_m,c_b\delta_b))(t_3-t_2)-{\rm
min}(\delta_m,c_b\delta_b)/2.
\end{equation}
To prove \eqref{keybigj}, first observe that under our assumptions
$P(t)-P(t-1)> -2\delta_m$ for all $t \in [t_0,t_0+1].$ This can be
verified using \eqref{bmass14} similarly to the proof of stability
in random mean reversion regime. Now
\begin{equation}\label{keybigj1}
P(t_3)-P(t_2) = \int_{t_2}^{t_3}\left(f(P(t),P_+)+\nu
S(P(t)-P(t-1))\right)\,dt + \sigma(B_{t_3}-B_{t_2}).
\end{equation}
If, on the other hand,
$P(t)-P_->a_b+\delta_m,$ then $P(t)-P(t-1)>a_b$ and therefore
$\nu S(P(t)-P(t-1)) > a_b +c_b\delta_b +\mu.$ Since
$f(P(t),P_+)>-\mu$ for all values of $P(t),$ we find that the
expression under integral in \eqref{keybigj1} is greater then
$a_b+c_b\delta_b.$ Now if $P(t)-P_- \leq a_b+\delta_m,$ then by
\eqref{bjump} we have
\[ P_+-P(t) > \frac{a_b+2\delta_m}{{\rm min}(\mu,1)}. \]
In this case the expression under integral in \eqref{keybigj1}
exceeds
\[ \mu\left(e^{\frac{a_b+2\delta_m}{{\rm min}(\mu,1)}}-1 \right)+\nu
S(-2\delta_m) \geq \mu \left(e^{\frac{a_b+2\delta_m}{{\rm
min}(\mu,1)}}-e^{\delta_m} \right) \geq \mu \frac{a_b+\delta_m}{{\rm
min(\mu,1)}}>a_b+\delta_m. \] Here in the first step we used $\nu
S(-2\delta_m)>\mu(1-e^{\delta_m}),$ which follows from Assumption
I'. Combining these estimates and assumption \eqref{bmass14} leads
to \eqref{keybigj}.

Next, suppose that the Brownian motion on the interval $[t_0+1,t_0+2]$
satisfies
\begin{equation}\label{bmass15}
\sigma(B_{t_3}-B_{t_2}) > -c_b\delta_b/2
\end{equation}
for any $t_3,t_2 \in [t_0+1,t_0+2],$ $t_3>t_2.$ Observe that
\eqref{keybigj} implies that $P(t_0+1)-P(t_0)>a_b.$ Let us show that
the inequality persists for all $t \in [t_0+1,t_0+2]$ given
\eqref{bmass15}. Assume, on the contrary, that $t_1 \in
[t_0+1,t_0+2]$ is the minimal time such that $P(t_1)-P(t_1-1)=a_b.$
Split $P(t_1)-P(t_1-1) = P(t_1)-P(t_0+1)+P(t_0+1)-P(t_1-1).$ Then
\[ P(t_1)-P(t_0+1) > (\nu S(a_b) -
\mu)(t_1-t_0-1)+\sigma(B_{t_1}-B_{t_0+1}) >
(a_b+c_b\delta_b)(t_1-t_0-1) -c_b\delta_b/2. \] On the other hand,
\eqref{keybigj} implies that
\[ P(t_0+1)-P(t_1-1) > (a_b+{\rm
min}(\delta_m,c_b\delta_b))(t_0+2-t_1)-{\rm
min}(\delta_m,c_b\delta_b)/2. \] Combining these two estimates, we
find that $P(t_1)-P(t_1-1)>a_b$ - contradiction.

It remains to observe that the probability of both \eqref{bmass14}
and \eqref{bmass15} happening can be estimated from below by
$2\Phi\left(\frac{{\rm min}(\delta_m,
c_b\delta_b)}{2\sqrt{2}\sigma}\right)-1.$
\end{proof}

\section{Numerical Simulations}\label{num}

We have tested the bubble equation dynamics on a wide range of
parameter values. The Matlab random number generator was
used for simulations.  The model seems to be quite robust to the
type of random forcing: mean zero uniformly distributed and normally
distributed random terms with comparable variances lead to very
similar results. This is again to be expected from the analytical
point of view, since the properties of the Brownian motion used in
the proofs are fairly generic and shared by many other random
processes. The simulations were performed with the nonlinear
response term $S(x) = \arctan(dx^{2n+1}),$ with $n=1,2$ or $3$ and
$d$ a varying parameter. The
roots $x_i,$ $i=1,...6$ were computed for each parameter set tried.
When clear bubble and collapse regimes were observable, the rates of
growth and decay showed very good agreement with the ones predicted
by values of the relevant roots.

As one can expect, the behavior of the model depends strongly on the
balance between different parameters. First of all, if we want to observe
bubbles, we should set $\nu \gtrsim \mu$. We found three
distinct regimes for the Bubble equation. \\

\begin{figure}[!ht]
\begin{center}
\scalebox{0.85}{\includegraphics{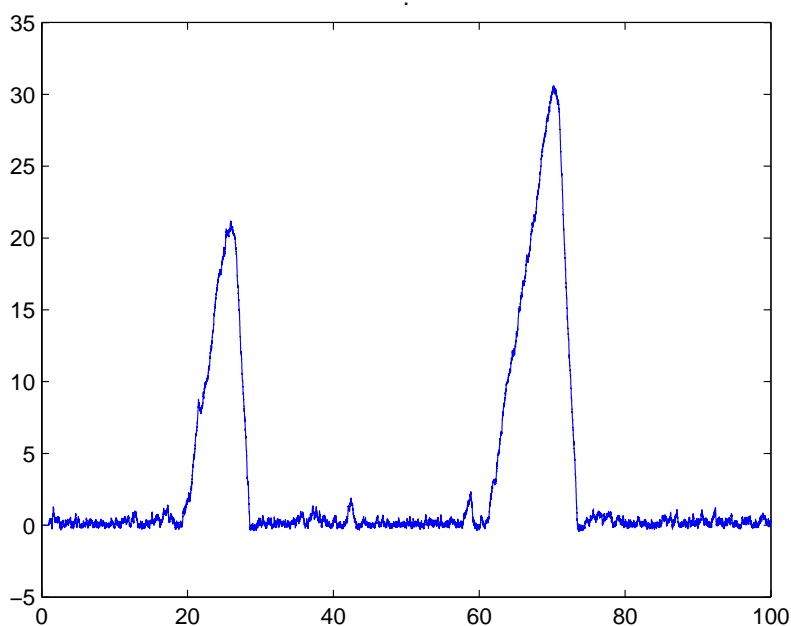}} \caption{Big
rare bubbles phase: $\nu \gtrsim \mu \gtrsim \sigma^2$}
\label{hrbub}
\end{center}
\end{figure}

\begin{figure}[!ht]
\begin{center}
\scalebox{0.85}{\includegraphics{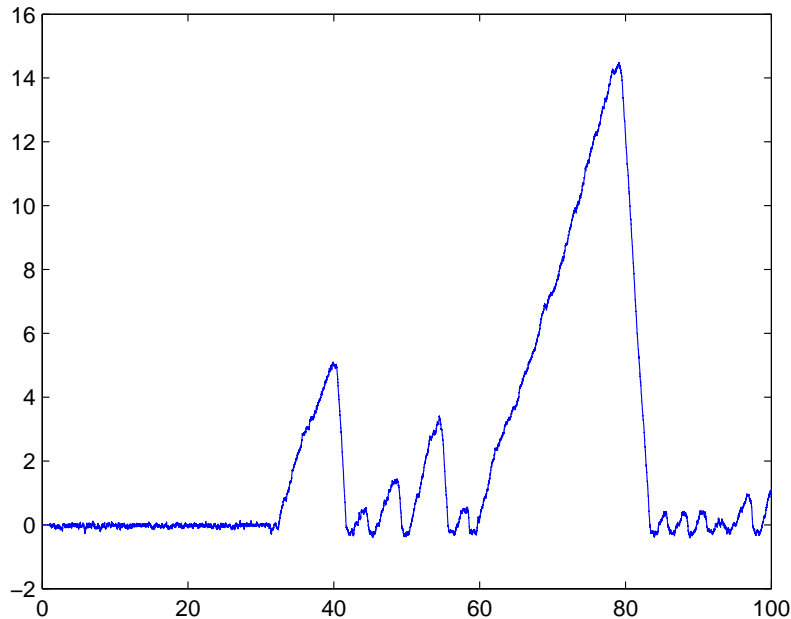}} \caption{Big
rare bubbles phase: $\nu \gtrsim \mu \gtrsim \sigma^2$: Serial
bubbles.} \label{unstbub}
\end{center}
\end{figure}

{\bf Regime 1: $\nu \gtrsim \mu \gtrsim \sigma^2$}, or big rare bubbles regime.
The three regimes described in analytical part of the paper are very
clear cut here. Transitions are relatively rare. The heuristic
picture of Section~\ref{her} seems to apply very well. A typical
simulation is shown on Figure~\ref{hrbub}. In this regime, we noted
that the likelihood of serial bubbles occurence, passing from
complete collapse right to the next bubble, may be significant. This
happens if the price compression below $P_0$ (which according
to the heuristics is of the order $\sim \log
\left(1+\frac{\delta_c}{\mu}\right)$) is comparable to $x_5.$ In that
case the recovery to $P_0$ level provides enough of a trend to help
ignite the next bubble (see Figure~\ref{unstbub}). 
The
run on Figure~\ref{hrbub} corresponds to $\nu = 5,$ $\mu = 4,$
$\sigma =3,$ $S$ cubic with $d=0.4.$ This corresponds to $x_1 \sim
-12,$ $x_2 \sim -7.5,$ $x_3 \sim -0.5,$ $x_4 \sim 0.5,$ $x_5 \sim
2,$ $x_6 \sim 3,$ and $\delta_m \sim 1.$ It is clear from these
values that the stability of the mean reversion and bubble regimes
is actually even stronger than one can expect from our analytical
bounds. In fact, taking smaller $\sigma$ usually led to very long
mean reversion runs, followed by eventual bubbles which grew huge
and exceeded our computer capacity before collapsing.

The fact that the dynamics in big rare bubbles regime is described
quite well by heuristics of Section~\ref{her} makes it also close to
the well known bubble model of Blanchard and Watson \cite{BW1982}.
It would be of interest to show that under certain scaling
assumptions, the dynamics of the bubble equation converges in the
rigorous sense to a Poisson-type process with just three regimes of
constant, exponential growth and exponential collapse. \\

\begin{figure}[!ht]
\begin{center}
\scalebox{0.85}{\includegraphics{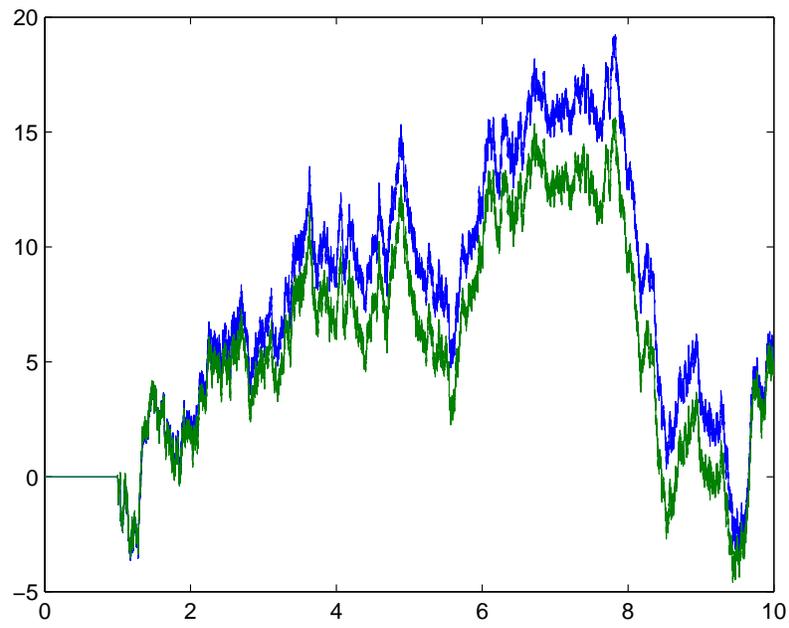}}
\caption{Strong randomness phase: $\sigma^2 \gtrsim \nu \gtrsim
\mu$} \label{randph}
\end{center}
\end{figure}
\clearpage

{\bf Regime 2: $\sigma^2 \gtrsim \nu \gtrsim \mu$}, or strong randomness regime.
In this phase the effects of the social response term are weak
compared to random fluctuations. The mean reversion, bubble and
collapse regimes are difficult to isolate, and the dynamics is far
from the heuristic model of Section~\ref{her}. A sample simulation
of this regime is shown on Figure~\ref{randph}. One graph is the
price evolution given by the bubble equation, and the other,
provided for comparison, is the dynamics corresponding to $\nu =0$
(and so it is just an Ornstein-Uhlenbeck process).

The bubble equation price tends to overshoot Ornstein-Uhlenbeck, but
fairly slightly, and there is no sustained bubble regime. The
parameters are $\nu = 0.6,$ $\mu = 0.2,$ $\sigma = 20,$ the function
$S$ is quintic with $d=90.$ \\

\begin{figure}[!ht]
\begin{center}
\scalebox{0.85}{\includegraphics{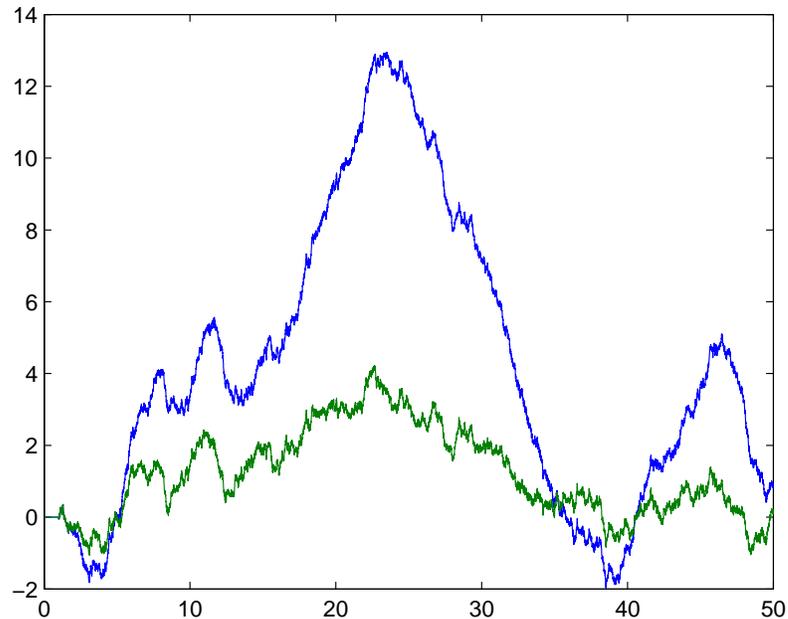}} \caption{The
balanced phase: $\sigma \sim \nu \gtrsim \mu$} \label{bal1}
\end{center}
\end{figure}

\begin{figure}[!ht]
\begin{center}
\scalebox{0.85}{\includegraphics{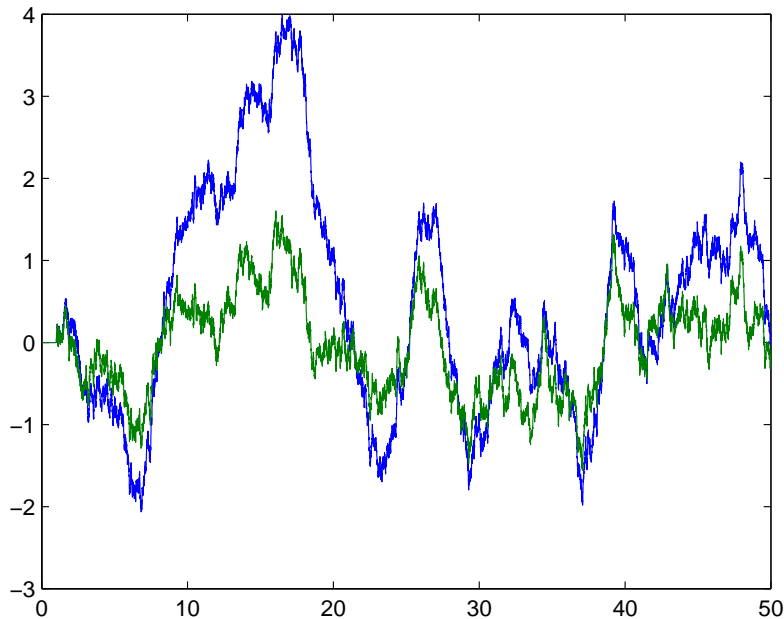}}
\caption{The balanced phase: stronger randomness} \label{bal2}
\end{center}
\end{figure}

\begin{figure}[!ht]
\begin{center}
\scalebox{0.85}{\includegraphics{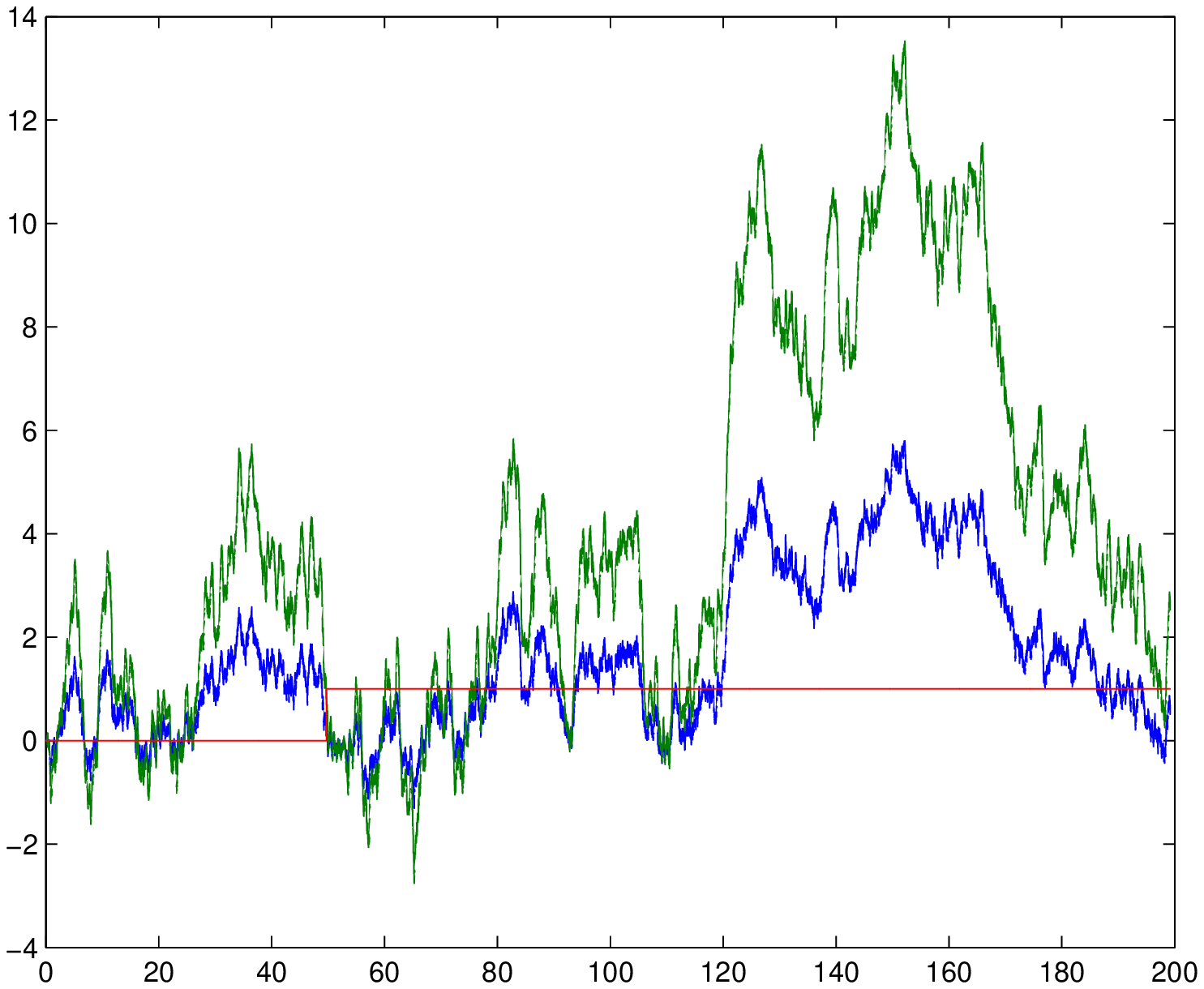}} \caption{A small jump in fundamentals:
no bubble ignition} \label{noj1}
\end{center}
\end{figure}

{\bf Regime 3.} The most interesting phase appears to be when $\sigma$ is
slightly larger or comparable to $\nu$, and $\mu$ is also comparable
or weaker. It appears that this is the
parameter range that is most likely to be relevant in most
applications with liquid instruments. Here one can still sometimes distinguish the different regimes
(although mean reversion is very unstable if $\mu$ is small). Yet
transitory effects are stronger than in big rare bubbles phase.
Figures~\ref{bal1} and \ref{bal2} show some typical runs in this
regime. As before, the Ornstein-Uhlenbeck process corresponding to the same
random sample but with $\nu=0$ is graphed for comparison.

Figure~\ref{bal1} corresponds to $\nu = 0.6,$ $\mu = 0.23$ and
$\sigma =2.$ The function $S$ is quintic, with $d=90.$ The
Figure~\ref{bal2} corresponds to $\nu =0.42,$ $\mu =0.15,$ $\sigma
=2,$ $S$ is cubic with $d=90.$

Figure~\ref{bal2} corresponds to relatively stronger randomness.
While the nonlinear effects are still strong, fairly long transitory
regimes become possible. Figure~\ref{bal1} shows an interesting
phenomenon of a collapse regime switching back into bubble
significantly above $P_0.$ It is clear by looking at the corresponding
Ornstein-Uhlenbeck process graph that randomness drives these switches. In
both simulations, social response term increases the effective
volatility of dynamics, stronger on the upside (creating bubbles)
but also on the downside after collapse (this difference in the strength
of the effect has clearly to do with
very quick growth of the mean reversion function once the price is
below $P_0$).

\begin{figure}[!ht]
\begin{center}
\scalebox{0.85}{\includegraphics{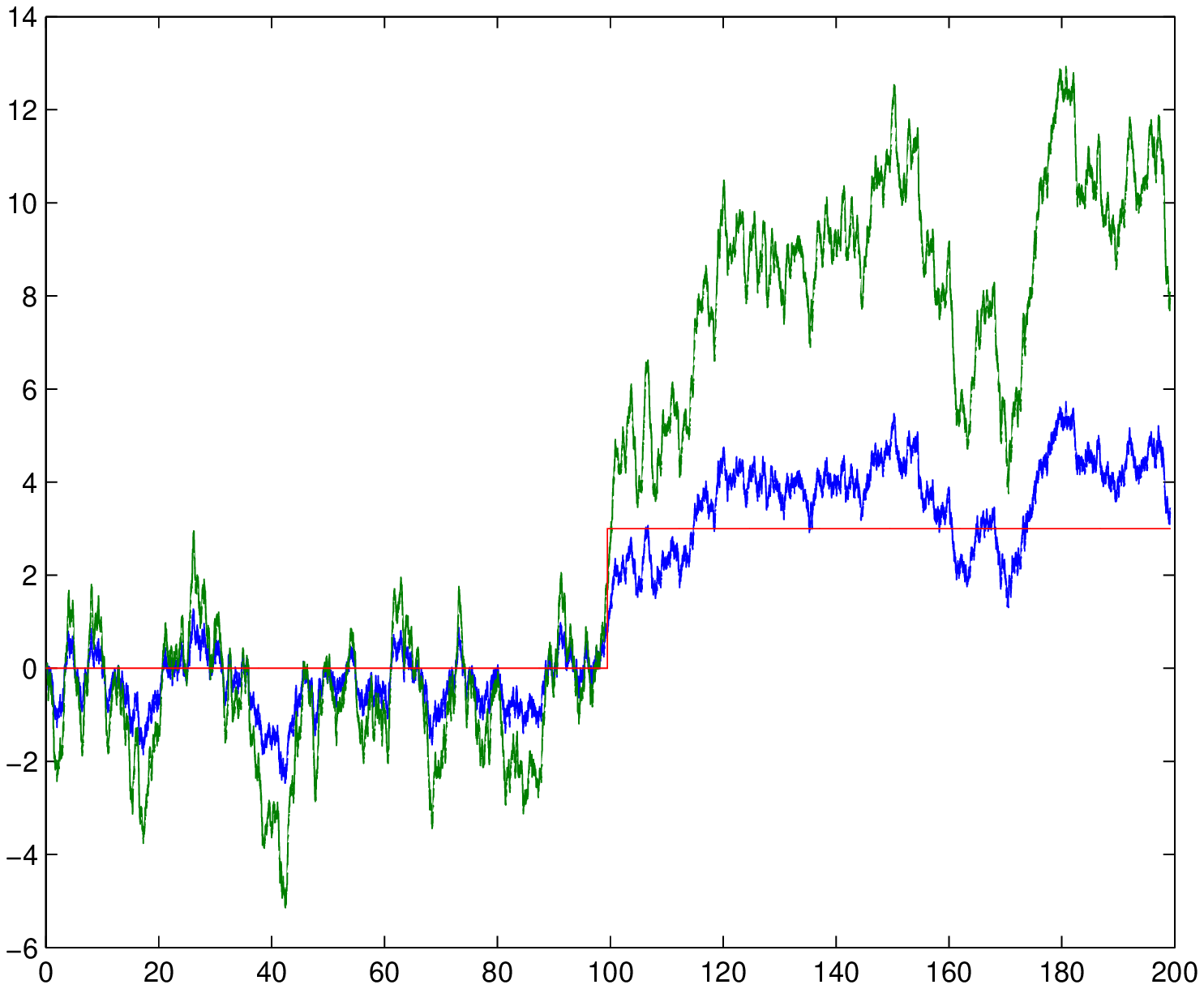}} \caption{A large jump in fundamentals causes
bubble ignition} \label{bj1}
\end{center}
\end{figure}

Next we investigate the likelihood of a jump in fundamentals igniting a bubble.
We find that the probability of the bubble ignition becomes significant if the
jump becomes comparable in size to $\nu.$ The Figure~\ref{noj1} shows a simulation where
the jump does not cause a bubble, while Figure~\ref{bj1} shows a larger jump causing
a bubble to ignite. On both figures, $P(t),$ $P_0(t)$ and the Ornstein-Uhlenbeck process
corresponding to $\nu=0$ case are graphed.

\begin{figure}[!ht]
\begin{center}
\scalebox{0.85}{\includegraphics{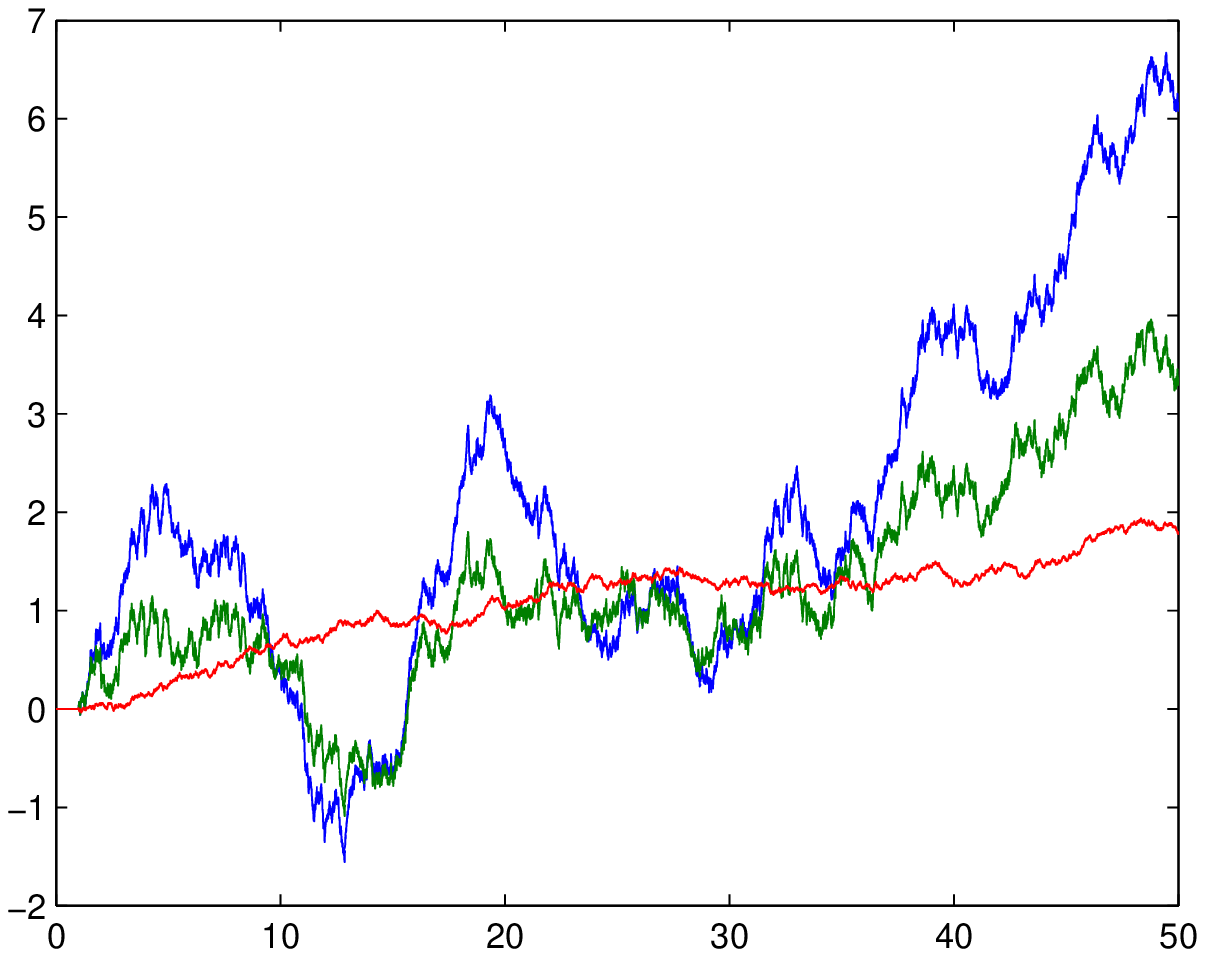}} \caption{The balanced
phase: varying $P_0$} \label{varp0}
\end{center}
\end{figure}

The final graph shows a simulation
where the stable value $P_0$ depends on time (Figure~\ref{varp0}).
The particular model we used is a random walk with a drift; the
volatility of $P_0$ is much lower than that of $P.$ All other
parameters of the bubble equation are the same as in the run
corresponding to Figure~\ref{bal2}. The graphs on Figure~\ref{varp0}
are those of $P_0(t),$ $P(t)$ and the corresponding Ornstein-Uhlenbeck process.

\section{Discussion and Further Directions}\label{disc}

The main goal of this paper was to present a simple
differential equations model of the effects of human psychology in
asset pricing. We tried to keep the model as simple as possible conditioned on the
fact that it should exhibit sufficiently rich set of behaviors.
A very interesting question that we did not address in this paper is the
statistical properties of price returns generated by \eqref{bubble}. It is well known
that empirical daily price returns of various financial instruments such as for example major indexes demonstrate deviations
from normal distribution, in particular power-like tails in the distribution of the returns.
Typically, lack of strong correlations in the ACF of daily returns is observed (even though
for some indexes and for some historical periods significant correlation can be present
(see e.g. \cite{LoMcK}). On the other hand, the volatility of returns exhibits clustering
with ACF decaying at a slow power rate. A number of models have been proposed that can fairly accurately
account for these statistical properties (see e.g. \cite{LM1,farjoshi,farlev}, and also \cite{SZSL,Sorn1}
for reviews). These models, however, tend to be more sophisticated than the one considered here.
Can one obtain similar results for a simpler model \eqref{bubble}? It is likely that some modifications
will be needed, in particular to reduce correlations in the return series.
We plan to discuss the statistical properties of the model
and some of its modifications in a future publication.

Another possible future direction of research involves bubble formation in
spatially distributed systems, such as the real estate markets.
The
price is now a function of both space and time, $P(x,t).$ In the
simplest case, the problem can be set on a graph, where each
vertex $v_j,$ $j=1,\dots, N$ corresponds to a town or a city. The
interaction between prices $p_i$ and $p_j$ in cities $v_i$ and $v_j$
is defined by a diffusion coefficient $\kappa_{ij}.$ This
coefficient models the degree of contact these cities have,
geographical proximity and generally the degree to which the real
estate prices in one city are likely to influence the prices in the
other. To each $v_i$ we also assign a number $q_i,$ which measures
the size of $v_i$'s market. In addition, the stable value $p^0_i$
varies from city to city. Then one possible model for price dynamics
is given by
\[ dp_i(t)= \sum_{j=1}^N \kappa_{ij}(p_i-p_j)\frac{1}{a_j}
+f(p_i,p^0_i)+\sigma dB^i_t+\nu S(p_i(t)-p_i(t-1)). \]
The relative simplicity of the equation \eqref{bubble} describing price in every
graph vertex (city) makes the model look approachable.
 Some of the interesting
questions to study in this case are possibility and likelihood of
bubble contagion (or bubble front propagation) doe to price diffusion, or possibility of
existence of localized bubbles.

Yet another interesting problem is suggested by numerical
simulations in the big rare bubbles regime. Is it true that under
certain scaling assumptions, the equation \eqref{bubble} can be rigorously
shown to converge to a nontrivial limit, such as suggested in
Section~\ref{her}? Does one get simply three possible clear cut regimes with
Poisson-like switching between them? A result like that would
establish an analytical link between relatively sophisticated bubble
equation and simpler discrete bubble models reminiscent of the one
appearing in Blanchard-Watson \cite{BW1982}.

\vspace{1cm}

{\bf Acknowledgement.} AK has been supported in part by NSF grant
DMS-0653813. LR has been supported in part by the NSF grant
DMS-0604687. AK thanks Igor Popov and Andrej Zlatos for useful
discussions. AK also thanks the University of Chicago for
hospitality.


\begin{thebibliography}{99}

\bibitem{AB} D.~Abreu and M.K.~Brunnermeier, \it Bubbles and
crashes, \rm Econometrica {\bf 71} (2003), 173--204.

\bibitem{AS92} M.C.~Adam and M.~Szafarz, \it Speculative bubbles and
financial markets, \rm Oxford Economic Papers, {\bf 44} (1992),
626--640.

\bibitem{BD} R.B.~Barsky and J.B.~Delong, \it Why does the stock
market fluctuate? \rm Quaterly Journal of Economics {\bf 107}
(1993), 291--311.

\bibitem{baum} W.J.~Baumol, \it Speculation, profitability and stability, \rm
Rev. Econ. Stat. {\bf 39}(1957), 263--271.


\bibitem{BW1982} O.J.~Blanchard and M.W.~Watson, \it Bubbles,
rational expectations and speculative markets, \rm in: Watchel P.,
eds., Crisis in Economic and Financial Structure: Bubbles, Bursts,
and Shocks. Lexington Books: Lexington, 1982.

\bibitem{BH} W.A.~Brock and C.H.~Hommes, \it Heterogenous beliefs
and routes to chaos in a simple asset pricing model, \rm Journal of
Economic Dynamics and Control {\bf 22} (1998), 1235--1274.

\bibitem{Cam89} C.~Camerer, \it Bubbles and Fads in Asset Prices,
\rm Journal of Ecomonic Surveys {\bf 3} (1989), 3--41.

\bibitem{Carswell} J.~Carswell, \it The South Sea Bubble, \rm London 1960: Cresset
Press.

\bibitem{Cassidy} J.~Cassidy, \it Dot.con: How America Lost its Mind and Its Money
in the Internet Era, \rm Harper Perennial, 2002.

\bibitem{Chancellor} E.~Chancellor, \it Devil take the hindmost: The history of
financial speculation, \rm Farrar, Straus and Giroux, 1999.


\bibitem{Cowles} V.~Cowles, \it The Great Swindle: The Story of the
South Sea Bubble, \rm New York (1960): Harper.

\bibitem{CH} A.M.G.~Cox and D.G.~Hobson, \it Local martingales,
bubbles and option prices, \rm Finance and Stochastics, {\bf 9}
(2005), 477--492.

\bibitem{Dash} M.~Dash, \it Tulipomania: The Story of the World's Most
Coveted Flower and the Extraordinary Passions It Aroused, \rm London
1999: Gollancz.

\bibitem{dhuang} R.H.~Day and W.~Huang, \it Bulls, bears, and market sheep, \rm
J. Econ. Behav. Organ. {\bf 14}(1990), 299--329.

\bibitem{D1} J.B.~Delong, A.~Scheifer, L.H.~Summers and
R.J.~Waldman, \it Noise trader risk in financial markets, \rm
Journal of Political Economy {\bf 98} (1990), 703--738.


\bibitem{Derman1} E.~Derman, \it The perception of time, risk and
return during periods of speculation, \rm Quantitative Finance {\bf
2} (2002), 282--296.

\bibitem{farmer} J.D.~Farmer, \it Market Force, Ecology and Evolution,
\rm Ind. and Corp. Change {\bf 11}(2002), 895--953.

\bibitem{farjoshi} J.D.~Farmer and S.~Joshi, \it
The Price Dynamics of Common Trading Strategies, \rm J. Econ. Behav. Organ. {\bf 49}(2002), 149--171.

\bibitem{Fuk98} Y.~Fukuta, \it A simple discrete-time approximation of continuous-time
bubbles, \rm Journal of Economic Dynamics and Control {\bf 22} (1998), 937--954.

\bibitem{Garber} P.~Garber, \it Famous First Bubbles: The Fundamentals
of Early Manias, \rm Cambridge 2000: MIT Press.

\bibitem{GH} A.~Gaunersdorfer and C.H.~Hommes, \it A nonlinear structural model for
volatility clustering, \rm Long Memory in Economics, G.~Teyssire and A.~Kirman, eds,
Berlin: Springer 2007, 265--288.

\bibitem{G1} C.~Gilles, \it Charges as equilibrium prices and as
bubbles, \rm Journal of Mathematical Economics {\bf 18} (1988),
155--167.

\bibitem{GL} C.~Gilles and S.F.~LeRoy, \it Bubbles and charges, \rm
International Economic Review, {\bf 33} (1992), 323--339.

\bibitem{Goldgar} A.~Goldgar, \it Tulipmania: Money, Honor, and Knowledge in
the Dutch Golden Age, \rm Chicago 2007: University of Chicago Press.

\bibitem{HLW} S.~Heston, M.~Loewenstein and G.A.~Willard, \it
Options and bubbles, \rm Reviews of financial studies, preprint.


\bibitem{IS} K.~Ide and D.~Sornette, \it Oscialltory finite-time singularities in
finance, population and rupture, \rm  Int. J. Mod. Phys. C {\bf 14}(2002), 267-275.

\bibitem{JM} R.A.~Jarrow and D.B.~Madan, \it Arbitrage, martingales,
and private monetary value, \rm Journal of Risk, {\bf 3} (2000).

\bibitem{JPS} R.A.~Jarrow, P.~Protter and K.~Shimbo, \it
Asset Price Bubbles in Incomplete Markets, \rm to appear in
Mathematical Finance.

\bibitem{Kind} C.~Kindleberger, \it Manias, Panics, and Crashes: A History
of Financial Crises, \rm Wiley, 2005.

\bibitem{KS08}  T.~Kaizoji and D.~Sornette, \it Market bubbles and
crashes, \rm  arXiv:0812.2449, long version of a shorter review
written for the Encyclopedia of Quantitative Finance (Wiley)



\bibitem{L} K.J.~Lansing, \it Lock-in of extrapolative expectations
in asset pricing model, \rm Macroeconomic Dynamics {\bf 10} (2006),
317--348.

\bibitem{Lo} A.~Lo, \it Reconciling Efficient Markets with Behavioral Finance:
The Adaptive Markets Hypothesis, \rm Journal of Investment Consulting {\bf 7}(2005), 21-44

\bibitem{LoMcK} A.~Lo and A. MacKinlay, \it A Non-Random Walk Down Wall Street,
\rm Princeton University Press, 1999

\bibitem{LW1} M.~Loewenstein and G.A.~Willard, \it Rational
equilibrim asset-pricing bubbles in continuous trading models, \rm
Journal of Economic Theory, {\bf 91} (2000), 17--58.

\bibitem{LW2} M.~Loewenstein and G.A.~Willard, \it Local
martingales, arbitrage and viability: free snacks and cheap thrills,
\rm Economic Theory, {\bf 16} (2000), 135--161.

\bibitem{lux} T.~Lux, \it Herd behavior, bubbles and crashes, \rm Econ. J. {\bf 105}(1995),
881--896.

\bibitem{LM1} T.~Lux and M.~Marchesi, \it Scaling and criticality in a stochastic
multi-agent model of a financial market, \rm Nature {\bf 397}(1999), 498--500.

\bibitem{LM2}  T.~Lux and M.~Marchesi, \it Volatility clustering in financial
markets: a micro-simulation of interacting agents, \rm Int. J. Theor. Appl. Finance {\bf 3}(2000),
675--702.

\bibitem{LS99} T.~Lux and D.~Sornette, \it On rational bubbles and
fat tails, \rm Journal of Money, Credit and Banking {\bf 34} (2002),
589--610.

\bibitem{McKay} C.~McKay, \it Extraordinary Popular Delusions and the Madness of
Crowds, \rm Harmony Books, 1980.


\bibitem{SZSL} E.~Samanidou, E.~Zschischang, D.~Stauffer and T.~Lux, \it
Agent-based models of financial markets, \rm Rep. Prog. Phys. {\bf 70}(2007), 409--450.

\bibitem{SX} J.A.~Scheinkman and W.~Xiong, \it Overconfidence and
speculative bubbles, \rm Journal of Political Economy {\bf 111}
(2003), 1183--1219.

\bibitem{Shiller} R.~Shiller, \it Irrational Exuberance, \rm
Princeton, NJ: Princeton University Press, 2005.

\bibitem{Shiratsuka} S.~Shiratsuka, \it Asset price bubble in Japan
in the 1980s: Lessons for Financial and Macroeconomic Stability, \rm
IMES Discussion Paper Series, Bank of Japan, paper 2003-E-15,
http://www.imes.boj.or.jp/english/publication/edps/2003/03-E-15.pdf.

\bibitem{Sorn1} D.~Sornette, \it Why stock markets crash: Critical
events in complex financial systems, \rm Princeton University Press, 2003

\bibitem{SorW} D.~Sornette and R.~Woodward, \it Financial bubbles, real estate
bubbles, derivative bubbles, and the financial and economic crisis, \rm
arXiv:math$/$0905.0220, 2009

\bibitem{farlev} S.~Thurner, J.D.~Farmer, and J.~Geanakoplos,  \it
Leverage Causes Fat Tails and Clustered Volatility, \rm preprint


\end{thebibliography}
\end{document}